\documentclass{lmcs}
\pdfoutput=1
\usepackage[utf8]{inputenc}

\usepackage{lastpage}
\lmcsdoi{21}{3}{27}
\lmcsheading{}{\pageref{LastPage}}{}{}%
{Nov.~02,~2021}{Sep.~12,~2025}{}

\usepackage{arydshln}
\usepackage{graphicx}   
\usepackage{hyperref}   
\hypersetup{
	colorlinks = true,
	citecolor=blue
}
\usepackage{thm-restate}
\usepackage{mathtools}
\usepackage{xspace}
\usepackage{verbatim}   
\usepackage{xcolor}
\usepackage{makecell}
\usepackage{tikz,pgfplots}
\usepackage{relsize}
\usepackage{booktabs}
\usepackage{tikz-qtree}
\usepackage{subcaption}
\makeatletter         
\def\figurecaption#1#2{\noindent\hangindent 40pt
                       \hbox to 36pt {\small\sl #1 \hfil}
                       \ignorespaces {\small #2}}

\long\def\@makecaption#1#2{
  \vskip 10pt 
  \settowidth{\@tempdima}{#2}
  \ifdim\@tempdima>0pt
       \setbox\@tempboxa\hbox{#1: #2}
     \else
       \setbox\@tempboxa\hbox{#1 #2}
   \fi
   \ifdim \wd\@tempboxa >\hsize               
       \begin{list}{#1:}{
       \settowidth{\labelwidth}{#1:}
       \setlength{\leftmargin}{\labelwidth}
       \addtolength{\leftmargin}{\labelsep}
        }\item #2 \end{list}\par   
     \else                                    
       \hbox to\hsize{\hfil\box\@tempboxa\hfil}  
   \fi}
\makeatother

\usepackage{multirow}
\usetikzlibrary{fit,shapes,trees,shapes.geometric}
\usetikzlibrary{patterns,decorations.pathreplacing,calc}
\usetikzlibrary{matrix, positioning, arrows}
\usetikzlibrary{chains,shapes.multipart}
\usepackage[linesnumbered,algoruled, lined, noend]{algorithm2e}
\usepackage{footnote}
\makesavenoteenv{algorithm}

\newcommand{\changes}[1]{{\color{black} {#1}}}

\newcommand{\reva}[1]{{\color{black}  #1}}

\newcommand{\revb}[1]{{\color{black}   #1}}

\newsavebox{\bullred}
\newsavebox{\bullolive}
\newsavebox{\bullblue}
\newsavebox{\bullteal}
\newsavebox{\bullbrown}
\newsavebox{\bullorange}
\sbox\bullred{\tikz{\draw[red,fill=red] circle (.3ex);}}
\sbox\bullolive{\tikz{\draw[olive,fill=olive] circle (.3ex);}}
\sbox\bullblue{\tikz{\draw[blue,fill=blue] circle (.3ex);}}
\sbox\bullteal{\tikz{\draw[teal,fill=teal] circle (.3ex);}}
\sbox\bullbrown{\tikz{\draw[brown,fill=brown] circle (.3ex);}}
\sbox\bullorange{\tikz{\draw[orange,fill=orange] circle (.3ex);}}

\makeatletter
\DeclareFontFamily{OMX}{MnSymbolE}{}
\DeclareSymbolFont{MnLargeSymbols}{OMX}{MnSymbolE}{m}{n}
\SetSymbolFont{MnLargeSymbols}{bold}{OMX}{MnSymbolE}{b}{n}
\DeclareFontShape{OMX}{MnSymbolE}{m}{n}{
	<-6>  MnSymbolE5
	<6-7>  MnSymbolE6
	<7-8>  MnSymbolE7
	<8-9>  MnSymbolE8
	<9-10> MnSymbolE9
	<10-12> MnSymbolE10
	<12->   MnSymbolE12
}{}
\DeclareFontShape{OMX}{MnSymbolE}{b}{n}{
	<-6>  MnSymbolE-Bold5
	<6-7>  MnSymbolE-Bold6
	<7-8>  MnSymbolE-Bold7
	<8-9>  MnSymbolE-Bold8
	<9-10> MnSymbolE-Bold9
	<10-12> MnSymbolE-Bold10
	<12->   MnSymbolE-Bold12
}{}

\let\llangle\@undefined
\let\rrangle\@undefined
\DeclareMathDelimiter{\llangle}{\mathopen}%
{MnLargeSymbols}{'164}{MnLargeSymbols}{'164}
\DeclareMathDelimiter{\rrangle}{\mathclose}%
{MnLargeSymbols}{'171}{MnLargeSymbols}{'171}
\makeatother

\usepackage[normalem]{ulem}

\newcommand{\mathsout}[1]
{\bgroup\mathchoice
	{\sbox0{$\displaystyle{#1}$}%
		\usebox0\hspace{-\wd0}%
		\rule[0.5\ht0-0.5\dp0-.5pt]{\wd0}{1pt}}%
	{\sbox0{$\textstyle{#1}$}%
		\usebox0\hspace{-\wd0}%
		\rule[0.5\ht0-0.5\dp0-.5pt]{\wd0}{1pt}}%
	{\sbox0{$\scriptstyle{#1}$}%
		\usebox0\hspace{-\wd0}%
		\rule[0.5\ht0-0.5\dp0-.5pt]{\wd0}{1pt}}%
	{\sbox0{$\scriptscriptstyle{#1}$}%
		\usebox0\hspace{-\wd0}%
		\rule[0.5\ht0-0.5\dp0-.5pt]{\wd0}{1pt}}%
	\egroup}

\newcommand{\cut}[1]{}   
\makeatletter
\newcommand{\mytag}[2]{%
	\text{#1}%
	\@bsphack
	\protected@write\@auxout{}%
	{\string\newlabel{#2}{{#1}{\thepage}}}%
	\@esphack
}
\makeatother

\newenvironment{packed_enum}{
	\begin{enumerate}
		\setlength{\itemsep}{1pt}
		\setlength{\parskip}{0pt}
		\setlength{\parsep}{0pt}
	}
	{\end{enumerate}}

\newlength\myboxwidth

\definecolor{code}{RGB}{230,230,230}

\newcommand{\introparagraph}[1]{\subparagraph{{\bf #1}}}  

\usepackage{aliascnt}

\providecommand{\bx}[0]{\mathbf{x}}

\providecommand{\bu}[0]{\mathbf{u}}

\providecommand{\mA}[0]{\mathcal{A}}

\providecommand{\mB}[0]{\mathcal{B}}

\providecommand{\mL}[0]{\mathcal{L}}

\providecommand{\tOUT}[0]{\texttt{OUT}}

\providecommand{\fhw}[0]{\texttt{fhw}}

\providecommand{\edges}[0]{\mathcal{E}}
\providecommand{\nodes}[0]{\mathcal{V}}

\providecommand{\domain}[0]{\mathbf{dom}}

\providecommand{\qstar}[1]{Q^*_{#1}}

\providecommand{\htree}[0]{\mathcal{T}}
\providecommand{\bag}[0]{\mathcal{B}}
\providecommand{\fhw}[1]{\mathsf{fhw}(#1)}

\providecommand{\mw}[0]{\mathsf{w}}

\providecommand{\eat}[1]{}

\AtBeginDocument{%
}

\tikzset{
	ncbar angle/.initial=90,
	ncbar/.style={
		to path=(\tikztostart)
		-- ($(\tikztostart)!#1!\pgfkeysvalueof{/tikz/ncbar angle}:(\tikztotarget)$)
		-- ($(\tikztotarget)!($(\tikztostart)!#1!\pgfkeysvalueof{/tikz/ncbar angle}:(\tikztotarget)$)!\pgfkeysvalueof{/tikz/ncbar angle}:(\tikztostart)$)
		-- (\tikztotarget)
	},
	ncbar/.default=0.25cm,
}

\tikzset{square left brace/.style={ncbar=0.25cm}}
\tikzset{square right brace/.style={ncbar=-0.25cm}}

\tikzset{round left paren/.style={ncbar=0.25cm,out=120,in=-120}}
\tikzset{round right paren/.style={ncbar=0.25cm,out=60,in=-60}}

\begin{document}

\title[Enumeration Algorithms for Conjunctive Queries with Projection]{Enumeration Algorithms for\texorpdfstring{\\}{} Conjunctive Queries with Projection}

\thanks{We are grateful to the reviewers for a careful reading of the manuscript and their feedback.}
    
	\author[S.~Deep]{Shaleen Deep\lmcsorcid{0000-0003-2342-4060}}[a]
	\address{Department of Computer Sciences, University of Wisconsin-Madison, Madison, Wisconsin, USA}
	\email{shaleen@cs.wisc.edu}

    \author[X.~Hu]{Xiao Hu\lmcsorcid{0000-0002-7890-665X}}[b]
    \address{David R. Cheriton School of Computer Science, University of Waterloo, Waterloo, Canada}
    \email{xiaohu@uwaterloo.ca}
	
	\author[P.~Koutris]{Paraschos Koutris\lmcsorcid{0000-0001-6309-1702}}[c]
	\address{Department of Computer Sciences, University of Wisconsin-Madison, Madison, Wisconsin, USA}
	\email{paris@cs.wisc.edu}

	\maketitle
		
	\begin{abstract}
	We investigate the enumeration of query results for an important subset of CQs with projections, namely star and path queries. The task is to design data structures and algorithms that allow for efficient enumeration with delay guarantees after a preprocessing phase. Our main contribution is a series of results based on the idea of interleaving precomputed output with further join processing to maintain delay guarantees, which maybe of independent interest. In particular, for star queries, we design combinatorial algorithms that provide instance-specific delay guarantees in  linear preprocessing time. These algorithms improve upon the currently best known results. Further, we show how existing results can be improved upon by using fast matrix multiplication. We also present new results involving tradeoff between preprocessing time and delay guarantees for enumeration of path queries that contain projections. Boolean matrix multiplication is an important query that can be expressed as a CQ with projection where the join attribute is projected away. Our results can therefore also be interpreted as sparse, output-sensitive matrix multiplication with delay guarantees.
\end{abstract}

	\section{Introduction}
\label{sec:intro}

The efficient evaluation of join queries over static databases is a fundamental problem in data management. There has been a long line of research on the design and analysis of algorithms that minimize the total runtime of query execution in terms of the input and output size~\cite{yannakakis1981algorithms,skewstrikesback,ngo2012worst}. However, in many data processing scenarios, it is beneficial to split query execution into two phases: the {\em preprocessing phase}, which computes a space-efficient intermediate data structure, and the {\em enumeration phase}, which uses the data structure to enumerate the query results as fast as possible, to minimize the {\em delay} between outputting two consecutive tuples in the result. This distinction is beneficial for several reasons. For instance, in many scenarios, the user wants to see one (or a few) results of the query as fast as possible: in this case, we want to minimize the time of the preprocessing phase such that we can output the first results quickly. On the other hand, a data processing pipeline may require that the result of a query is accessed multiple times by a downstream task. In this case, spending more time during the preprocessing phase is better to guarantee a faster enumeration with a smaller delay. 
 
Previous work in the database literature has focused on finding the class of queries that can be computed with $O(|D|)$ preprocessing time (where $D$ is the input database instance) and constant delay during the enumeration phase. The main result in this line of work shows that full (i.e., without projections) acyclic Conjunctive Queries (CQs) admit linear preprocessing time and constant delay~\cite{bagan2007acyclic}. If the CQ is not full but its free variables satisfy the {\em free-connex} property, the same preprocessing time and delay guarantees can still be achieved. It is also known that for any (possibly non-full) acyclic CQ, it is possible to achieve linear delay after linear preprocessing time~\cite{bagan2007acyclic}. Prior work that used structural decomposition methods~\cite{greco2013structural} generalized these results to arbitrary CQs with free variables and showed that the projected solutions could be enumerated with $O(|D|^\fhw)$ delay, where   $\fhw$ is the {\em fractional hypertree width}~\cite{gottlob2014treewidth} of the query. Moreover, a dichotomy about the classes of conjunctive queries with fixed arities where such answers can be computed with polynomial delay (WPD) is also shown. When the CQ is full but not acyclic, factorized databases use $O(|D|^{\fhw})$ preprocessing time to achieve constant delay. 
 We should note here that we can always compute and materialize the query result during preprocessing to achieve constant delay enumeration, but this requires an exponential amount of space.

The aforementioned prior work investigates specific points in the preprocessing time-delay tradeoff space. While the story for full acyclic CQs is relatively complete, the same is not true for general CQs, even for acyclic CQs with projections. For instance, consider the simplest such query: $Q_{\texttt{two-path}} = \pi_{x,z} (R(x,y) {\Join} S(y,z))$, which joins two binary relations and then projects out the join attribute. For this query, \cite{bagan2007acyclic} ruled out a constant delay algorithm with linear time preprocessing unless the {boolean} matrix multiplication exponent is $\omega = 2$. However, we can obtain $O(|D|)$ delay with $O(|D|)$ preprocessing time. We can also obtain $O(1)$ delay with $O(|D|^2)$ preprocessing by computing and storing the full result. It is worth asking whether there are other interesting points in this tradeoff between preprocessing time and delay. 
Towards this end, seminal work by Kara et al.~\cite{kara19} showed that for any hierarchical CQ\footnote{Hierarchical CQs are a strict subset of acyclic CQs.} (possibly with projections), there always exists a smooth tradeoff between preprocessing time and delay. This is the first improvement over the results of Bagan et al.~\cite{bagan2007acyclic} in over a decade for queries involving projections. Applied to the query $Q_{\texttt{two-path}}$, the main result of~\cite{kara19} shows that for any $\epsilon \in [0,1]$, we can obtain $O(|D|^{1-\epsilon})$ delay with $O(|D|^{1+\epsilon})$ preprocessing time. 

In this article, we continue the investigation of the tradeoff between preprocessing time and delay for CQs with projections. We focus on two classes of CQs: {\em star queries}, a popular subset of hierarchical queries, and a useful subset of non-hierarchical queries known as {\em path queries}. We focus narrowly on these two classes for two reasons. First, star queries are of immense practical interest given their connections to set intersection, set similarity joins, and applications to entity matching (we refer the reader to~\cite{deep2020fast} for an overview). The most common star query seen in practice is $Q_{\texttt{two-path}}$~\cite{bonifati2020analytical}. The same holds for path queries, which are fundamental in graph processing. Second, as we will see in this article, the tradeoff landscape is complex and requires developing novel techniques even for the simple class of star queries. We also present a result on another subset of hierarchical CQs called left-deep. Our key insight is to design enumeration algorithms that depend not only on the input size $|D|$ but are also aware of other data-specific parameters, such as the output size. To give a flavor of our results, consider $Q_{\texttt{two-path}}$ as an example, and denote by $\tOUT_{\Join}$ the output of the underlying full query (i.e., $R(x,y) {\Join} S(y,z)$). We can show the following result:

\begin{thm} \label{thm:path2}
Given any database instance $D$, after $O(|D|)$ preprocessing time,  the output of $Q_{\texttt{\upshape two-path}} =  \pi_{x,z} (R(x,y) {\Join} S(y,z))$ can be enumerated with $O(|D|^2/ |\tOUT_{\Join}|)$ delay.
\end{thm}

At this point, the reader may wonder about the improvement from the above result.~\cite{kara19} implies that with preprocessing time $O(|D|)$, the delay guarantee in the worst-case is $O(|D|)$. This raises the question of whether the delay from Theorem~\ref{thm:path2} is truly an algorithmic improvement rather than an improved analysis of~\cite{kara19}. We answer the question positively. Specifically, we show that there exists a database instance where the delay obtained from Theorem~\ref{thm:path2} is a polynomial improvement over the actual guarantee~\cite{kara19} and not just the worst-case. \reva{Our proposed algorithm outperforms prior work by cleverly using the stored precomputed output to maintain the delay guarantee and alternate with further join processing that generates new output but is not fast enough to satisfy the delay guarantee}. When the preprocessing time is linear, the delay implied by our result depends on the size of the full join. For the worst-case output size where $|\tOUT_{\Join}| = \Theta(|D|^2)$, we obtain the best possible delay, which will be constant. Compare this to the result of~\cite{kara19}, 
which would require nearly $O(|D|^2)$ preprocessing time to achieve the same guarantee. On the other hand, if $|\tOUT_{\Join}| = \Theta(|D|)$, we obtain only a linear delay guarantee of $O(|D|)$.\footnote{We do not need to consider the case where $|\tOUT_{\Join}| \leq |D|$, since then we can materialize the full result during the preprocessing time {using constant delay enumeration for queries without projections~\cite{olteanu2015size}}.} The reader may wonder how our result compares in general with the tradeoff in~\cite{kara19} \changes{in the worst-case}; we will show that we can always get at least as good of a tradeoff point as the one in~\cite{kara19}.~\autoref{table:results} summarizes the prior work and the results present in this article.

\begin{figure*}
	\small
\begin{tabular}{lll@{}r}
	\toprule
	\textbf{Queries} & \textbf{Preprocessing}  & \textbf{Delay}  & \textbf{Source} \\
	\toprule
	{Arbitrary} acyclic CQ & $O(|D|)$ & $O(|D|)$  &  \cite{bagan2007acyclic} \\ \midrule
	\makecell[l]{Free-connex CQ (projections)} & $O(|D|)$ & $O(1)$   &  \cite{bagan2007acyclic} \\ \midrule
	Full CQ & $O(|D|^\fhw)$ & $O(1)$  & \cite{olteanu2016factorized} \\ \midrule
    \reva{ Arbitrary CQ with aggregations} & \reva{${O}(|D|^\texttt{faqw}  \log |D|)$} & \reva{$O(1)$} & \reva{\cite{abo2016faq} + \cite{olteanu2016factorized}} \\ \midrule
	Full CQ & ${O}(|D|^\texttt{subw} \log |D|)$ & $O(1)$  & \cite{abo2017shannon} + \cite{olteanu2016factorized} \\ \midrule

	\makecell[l]{Hierarchical CQ \\ (with projections)} & $O(|D|^{1+(\mw-1) \epsilon})$ & \makecell[l]{$O(|D|^{1-\epsilon})$, $\epsilon \in [0,1]$}   & \cite{kara19} \\ \midrule
	\makecell[l]{Star query with $k$  relations \\ (with projections)} & $O(|D|)$ & $O(\frac{|D|^{k/(k-1)}}{ |\tOUT_{\Join}|^{1/(k-1)}})$  &{\color{red} this article} \\ \midrule
	\makecell[l]{Path query with $k$ relations \\ (with projections)} & $O(|D|^{2-\epsilon/(k-1)})$ & \makecell[l]{$O(|D|^\epsilon),$ $ \epsilon \in [0,1)$}  & {\color{red} this article} \\	 \midrule
	\makecell[l]{Left-deep hierarchical CQ  \\ (with projections)} & $O(|D|)$ & $O(|D|^k/|\tOUT_{\Join}|)$ & {\color{red} this article} \\	\midrule
	\makecell[l]{Two path query \\ (with projections)} & $O(|D|^{\omega \cdot \epsilon})$ & \makecell[l]{$O(|D|^{1-\epsilon}),$ $ \epsilon \in [\frac{2}{\omega+1}, 1]$}  
	& {\color{red} this article} \\			
	\bottomrule
\end{tabular}
\caption{Preprocessing time and delay guarantees for different queries. $|\tOUT_{\Join}|$ denotes the size of join query under consideration but without any projections. \texttt{subw} denotes the submodular width of the query~\cite{marx2013tractable} and \texttt{faqw} denote the \textsf{FAQ}-width of the query~\cite{grohe2014constraint}. If the CQ contains no aggregations, it is known that $\texttt{faqw} = \texttt{fhw}$. {For each class of query, the total running time is $O(\min\{|D| \cdot |\tOUT_{\pi}|, |D|^\mathtt{subw} \log |D| + |\tOUT_{\pi}|\})$ where $|\tOUT_{\pi}|$ denotes the size of the query result.}} \label{table:results}
\end{figure*} 

\introparagraph{Our Contribution} In this article, we improve the state-of-the-art preprocessing time-delay tradeoff for a subset of CQs with projections. We summarize our main technical contributions below (highlighted in~\autoref{table:results}):

\begin{enumerate}

\item Our first contribution is a novel algorithm (Theorem~\ref{thm:star:delay}) that achieves output-dependent delay guarantees for star queries (formally defined in Section~\ref{subsec:cq}) after linear preprocessing time. Specifically, 
we can achieve $O(|D|^{k/(k-1)}/ |\tOUT_{\Join}|^{1/k-1})$ delay after linear preprocessing.
Our key idea is to identify an appropriate degree threshold to split a relation into partitions of {\em heavy} and {\em light}, which allows us to perform efficient enumeration. For star queries, our result implies no smooth tradeoff between preprocessing time and delay guarantees as stated in~\cite{kara19} for the general class of {hierarchical} queries. 

\item We introduce the novel idea of {\em interleaving} join query computation in the context of enumeration algorithms, which forms the foundation for our algorithms and may be of independent interest. Specifically, we show that it is possible to union the output of two algorithms $\mA$ and $\mA'$ with $\delta$ delay guarantee where $\mA$ enumerates query results with $\delta$ delay guarantees but $\mA'$ does not. This technique allows us to compute a subset of a query {\em on the fly} when enumeration with good delay guarantees is impossible.

\item We use fast matrix multiplication techniques to obtain a better preprocessing time-delay tradeoff than~\cite{kara19}.  We also show an algorithm for left-deep hierarchical queries with linear preprocessing time and output-dependent delay guarantees.

\item Finally, we present new results on preprocessing time-delay tradeoffs for non-hierarchical CQs, the class of path queries (formally defined in Section~\ref{subsec:cq}). 
Our results show that we can achieve delay $O(|D|^\epsilon)$ with preprocessing time $O(|D|^{2-\epsilon/(k-1)})$ for any $\epsilon \in [0,1)$.
\end{enumerate}

This article is the full version of a conference publication~\cite{deep2021enumeration}. We have added all of the proofs and intermediate results excluded from the paper. In particular, we have added the full proof of our main result -- enumeration delay obtained for star queries (Theorem~\ref{thm:star:delay}, Theorem~\ref{thm:main:leftdeep}, and Theorem~\ref{thm:path}) as well as the proofs of all the helper lemmas in~\autoref{sec:helper}. Additionally, we have added the full algorithm for the helper lemmas. We have also added Example~\ref{ex:quadratic} and Example~\ref{ex:enum:two} to improve the exposition. Finally, we have added a new result in~\autoref{subsec:interleaving} that transfers a result for static enumeration to the dynamic setting.

The remainder of the article is organized as follows. In~\autoref{sec:framework}, we give
preliminary concepts, definitions, and notation. In~\autoref{sec:helper}, we present useful lemmas used frequently throughout the article. In~\autoref{sec:proof}, we state the main result and compare it against prior work. As a warm-up step, we present the intuition and main ideas for the two-path query, which are then generalized to star queries. Toward the end of the section, we also discuss how matrix multiplication can be useful.~\autoref{sec:ldeep} and~\autoref{sec:path} are dedicated to studying a class of hierarchical queries that we call left-deep and path queries, respectively. Finally, we discuss past contribution measures in~\autoref{sec:related} and conclude in~\autoref{sec:conclusion}. 

	\section{Problem Setting}
\label{sec:framework}
In this section, we present the basic notation and terminology.
\subsection{Conjunctive Queries}
\label{subsec:cq}
We will focus on the class of {\em conjunctive queries} (CQs):
$$ \label{eq:q}
Q = \pi_{\mathbf{y}}(R_1(\bx_1) \Join R_2(\bx_2) \Join \ldots \Join R_n(\bx_n) )
$$
Here, the symbols $\mathbf{y},\bx_1, \dots, \bx_n$ are vectors that contain {\em variables} or {\em constants}. We say that $Q$ is {\em full} if there is no projection. We will typically use the symbols 
$x,y,z,\dots$ to denote variables, and $a,b,c,\dots$ to denote constants.
We use $Q(D)$ to denote the query's result $Q$ over input database $D$. In this paper, we will focus on CQs with no constants and no repeated variables in the same atom (both cases can be handled within a linear time preprocessing step, so this assumption is without any loss of generality). Such a query can be represented equivalently as a {\em hypergraph}  $\mathcal{H}_Q = (\nodes_Q, \edges_Q)$, where $\nodes_Q$ is the set of variables, and for each hyperedge $F \in \edges_Q$ there exists a relation $R_F$ with variables $F$. We will be particularly interested in two families of CQs fundamental to query processing: star and path queries. The {\em star query} with $k$ relations is expressed as:
$$
\qstar{k} = \pi_{x_1,x_2,\dots, x_k}R_1(x_1,y) \Join R_2(x_2, y) \Join \dots \Join R_k(x_k, y)
$$
where $x_1, \dots, x_k$ are distinct variables. 
The {\em path query} with $k$ relations is expressed as:
$$
P_{k} = \pi_{x_1, x_{k+1}}R_1(x_1, x_2) \Join R_2(x_2, x_3) \Join \dots \Join R_k(x_k, x_{k+1})
$$

\smallskip
\noindent {\bf Hierarchical Queries.}
 A CQ $Q$ is {\em hierarchical} if for any two of its variables, either {the sets of atoms} in which they occur are disjoint or one is contained in the other~\cite{suciu2011probabilistic}. For example, $\qstar{k}$ is hierarchical for any $k$, while $P_k$ is hierarchical only when $k \leq 2$.
 
\introparagraph{Join Size Bounds}
Let $\mathcal{H} = (\nodes, \edges)$ be a hypergraph, and $S \subseteq \nodes$.
A weight assignment $\bu = (u_F)_{F \in \edges}$ 
is called a {\em fractional edge cover} of $S$ if 
$(i)$ for every $F \in \edges, u_F \geq 0$  and $(ii)$ for every
$x \in S, \sum_{F:x \in F} u_F \geq 1$. 
The {\em fractional edge cover number} of $S$, denoted by
$\rho^*_{\mathcal{H}}(S)$ is the minimum of
$\sum_{F \in \edges} u_F$ over all fractional edge covers of $S$. 
We write $\rho^*(\mathcal{H}) = \rho^*_{\mathcal{H}}(\nodes)$.

\introparagraph{Tree Decompositions}
Let $\mathcal{H} = (\nodes, \edges)$ be a hypergraph of a CQ $Q$.
A {\em tree decomposition} of $\mathcal{H}$ is  a tuple 
$(\htree, (\bag_t)_{t \in V(\htree)})$ where $\htree$ is a tree, and 
every $\bag_t$ is a subset of $\nodes$, called the {\em bag} of $t$, such that 
	\begin{packed_enum}
		\item  
		each edge in $\edges$ is contained in some bag; and
		\item 
		for each variable $x \in \nodes$, the set of nodes $\{t \mid x \in \bag_t\}$ form a connected subtree of $\htree$.
	\end{packed_enum}

The {\em width} of a tree decomposition $(\htree, (\bag_t)_{t \in V(\htree)})$ is 
defined as $\max_{t \in V(\htree)} \rho^*(\bag_t)$, where
$\rho^*(\bag_t)$ is the minimum fractional edge cover of the vertices in $\bag_t$.
The {\em fractional hypertree width} of a CQ $Q$, denoted as $\fhw(Q)$, is the minimum width among all tree decompositions of its hypergraph.
A query is {\em acyclic} if $\fhw(Q)=1$.

\introparagraph{Computational Model}
To measure the running time of our algorithms, we will use the uniform-cost RAM 
model~\cite{hopcroft1975design}, where data values and pointers are of $O(1)$ size. All complexity results concern data complexity, where the query is assumed to be fixed. \revb{We will assume that a hash table or a hash map $H$ can insert and delete a key in constant time. Further, given a key $u$, it can check existence, and return the value for the key (denoted by $H[u]$) in constant time. In practice, hashing can only achieve amortized constant time for some of the operations~\cite{cormen2022introduction}. Therefore, in the paper, whenever we claim constant time operations for hash table, we mean amortized constant time.}

\subsection{Fast Matrix Multiplication}
\revb{The matrix multiplication exponent $\omega$ is the smallest number such that for any $\epsilon > 0$, there is an algorithm that multiplies two rational $n \times n$ matrices with at most $O(n^{\omega + \epsilon})$ arithmetic operations.} Let $A$ be a $U \times V$ matrix and $C$ be a $V \times W$ matrix over any field $\mathcal{F}$. $A_{i,j}$ is the shorthand notation for entry of $A$ located in row $i$  and column $j$. The matrix product is given by $(AC)_{i,j} = \sum_{k=1}^{V} A_{i,k} C_{k,j}$. Given their fundamental importance, algorithms for fast matrix multiplication are of great theoretical interest. We will use the following folklore lemma about rectangular matrix multiplication.

\begin{lem} \label{lem:matrix:multiplication}
	\revb{Let $\omega$ be the matrix multiplication exponent.} Let $\beta = \min\{U,V,W\}$. Fast matrix multiplication of Boolean matrices of size $U \times V$ and $V \times W$ runs in time $O(UVW \beta^{\omega-3})$.
\end{lem}

In Lemma~\ref{lem:matrix:multiplication}, matrix multiplication cost dominates the time required to construct the input matrices (if they have not been built already) for all $\omega \geq 2$. Fixing $\omega = 2$, rectangular matrix multiplication can be done in time $O(UVW/\beta)$.  
A long line of research on fast {square} matrix multiplication has dropped the complexity to $O(n^\omega)$, where $2 \leq \omega< 3$. The current best-known value is $\omega = 2.3715$~\cite{williams2024new}, but it is believed that the actual value is $2$, which would imply that the Boolean matrix product can be computed in time $O(n^{2 + o(1)})$. \revb{We refer the reader to Section 6 in~\cite{berkholz2020constant} for a discussion on the topic.}

\subsection{Problem Statement}
\label{subsec:ps}

Given a CQ $Q$ and an input database $D$, we want to enumerate the tuples in $Q(D)$ in any order.
We will study this problem in the enumeration framework similar to that of~\cite{Segoufin15}, where an algorithm can be decomposed into two phases:

\begin{itemize}
	\item {\bf Preprocessing:} it builds a data structure that takes space {\em $S_p$} in {\em preprocessing time} $T_p$.
	\item {\bf Enumeration:} it outputs $Q(D)$ with no repetitions. This phase has access to any data structures constructed in the preprocessing phase and can also use additional space of size $S_e$.
	The {\em delay} $\delta$ is the maximum time duration between outputting any pair of consecutive tuples (this includes the time to output the first tuple and the time to notify that the enumeration phase has been completed).  
\end{itemize}

We aim to study the tradeoff between the preprocessing time $T_p$ and delay $\delta$ for a given CQ $Q$. We want to achieve a constant delay in linear preprocessing time. As Table~\ref{table:results} shows, when $Q$ is full, after $O(|D|^\fhw)$ preprocessing time, we can enumerate the results with constant delay~\cite{olteanu2016factorized}. When $Q$ is acyclic, i.e., $\fhw=1$, we can achieve a constant delay with only linear preprocessing time. On the other hand,~\cite{bagan2007acyclic} shows that for every acyclic CQ, we can achieve a linear delay with linear preprocessing time.
\cut{When $Q$ has no projections, several results are known about enumeration algorithms. In particular, it is known that with, \footnote{The preprocessing time can be further improved to $T_p = O(|D|^\texttt{subw})$ where \texttt{subw} denotes the submodular width of the query.}. Here, $\fhw$ is the fractional hypertree width of $Q$. 
Much less is known about enumeration algorithms for joins followed by a projection. If $Q$ is {\em free-connex acyclic}, \cite{bagan2007acyclic} shows that enumeration can still be done with linear preprocessing time and constant delay. However, not all acyclic queries are free-connex acyclic.
For such queries, the simplest way to evaluate is to first compute the full join result, then deduplicate, and finally enumerate the result. This guarantees constant delay enumeration, but the preprocessing phase is prohibitively expensive, since it must compute the full join result, which in the worst case can be as large as $|D|^{\rho^*}$. For instance, suppose we want to compute the query $\pi_{\bx_1, \dots, \bx_k}(\qstar{k})$. In this case, $\rho^* = k$, and hence we would obtain preprocessing time $T_p = O(|D|^k)$. Note that, in contrast, since $\qstar{k}$ is acyclic, we can enumerate the full result (without projections) with $T_p = O(|D|)$ and constant delay. On the other extreme, . }
Recently,~\cite{kara19} showed that it is possible to get a tradeoff between the two extremes for the class of hierarchical queries, which are acyclic but not necessarily free-connex. This is the first non-trivial result that improves upon the linear delay guarantees given by~\cite{bagan2007acyclic} for CQs with projections.

\begin{thmC}[\cite{kara19}] \label{thm:basic:three}
	Consider a hierarchical CQ $Q$ with factorization width $\mw$, and an input instance $D$. Then, for any $\epsilon \in [0,1]$ there exists an algorithm that can preprocess $D$ in time $T_p =  {O}(|D|^{1+(\mw-1)\epsilon})$ and space $S_p = {O}(|D|^{1+(\mw-1)\epsilon})$ such that we can enumerate the query output with
	$\delta = O(|D|^{1-\epsilon})$ delay and $S_e = O(1)$ space.
\end{thmC}
 
{The factorization width $\mw$ of a CQ, originally introduced as $s^\uparrow$~\cite{olteanu2015size}, is a generalization of the $\fhw$ from Boolean to arbitrary CQs.} For a star query $\qstar{k}$, $\mw = k$.{ Observe that preprocessing time $T_p$ must always be smaller than the time required to evaluate the full join result. This is because if $T_p = \Theta(|\tOUT_{\Join}|)$, we can evaluate the full join and deduplicate the projection output, allowing us to obtain constant delay in the enumeration phase. Therefore, the tradeoff is more meaningful when $\epsilon$ can only take values between $0$ and $(\log_{|D|} |\tOUT_{\Join}| -1) / (\mw-1)$. In the worst-case, $|\tOUT_{\Join}| = |D|^\mw$ and $\epsilon$ can take any value between $0$ and $1$ (both inclusive), which is captured by the general result above.}
	
	\section{Helper Lemmas}
\label{sec:helper}

Before we present the proof of our main results, we discuss three useful lemmas which will be used frequently and may be of independent interest for enumeration algorithms. The first two lemmas are based on the key idea of \emph{interleaving query results}, which we describe next. {We note that the idea of interleaving computation has been explored in the past to develop dynamic algorithms with good worst-case bounds using static data structures~\cite{overmars1981dynamization}.} 

We say that an algorithm $\mA$ provides no delay guarantees, meaning its delay guarantee can be as large as its total execution time. In other words, if an algorithm requires time $T$ to complete, its delay guarantee is upper bound by $T$. Since we use the uniform-cost RAM model, each operation takes one unit of time.
 
\begin{lem}\label{lem:helper:one}
	Consider two algorithms $\mA, \mA'$ and two positive integers $T$ and $T'$ provided as a part of the input such that
	\begin{enumerate}
		\item $\mA$ enumerates query results in total time at most {$T$} with no delay guarantees.
		\item $\mA'$ enumerates query results with delay $\delta$ and runs in total time at least {$T'$}.
		\item The outputs of $\mA$ and $\mA'$ are disjoint.
	\end{enumerate}
	Then, the union of the outputs of $\mA$ and $\mA'$ can be enumerated with delay $c \cdot \delta \cdot \max\{1, T/T'\}$ for some constant $c$.
\end{lem}
\begin{proof}
	Let $\eta$ and $\gamma$ denote two positive values that will be fixed later. Note that after $\delta$ time has passed, we can emit one output result from $\mA'$. But since we also want to compute the output from $\mA$ that takes overall time $T$, we need to slow down the enumeration of $\mA'$ sufficiently so we do not run out of output from $\mA'$. This can be done by interleaving the two algorithms in the following way:  we run $\mA'$ for $\gamma$ operations, pause $\mA'$, then run $\mA$ for $\eta$ operations, pause $\mA$ and resume $\mA'$ for $\gamma$ operations, and so on. The pause and resume takes constant time (say $c_{\textsf{pause}}$ and $c_{\textsf{resume}}$) in RAM model where the state of registers and program counter can be stored and retrieved enabling pause and resume of any algorithm. Our goal is to find a value of $\eta$ and $\gamma$ such that $\mA'$ does not terminate until $\mA$ has finished. This condition is satisfied when the number of iterations of $\mA'$ is equal to the number of iterations of $\mA$. This gives us the condition that,
	$$ T'/\gamma \leq \text{(Time taken by }\mA'\text{)} / \gamma = \text{(Time taken by }\mA\text{)} / \eta \leq T/\eta $$
	
	Thus, any value of $\eta \leq T \cdot \gamma / T'$ is acceptable. We fix $\eta$ to be any positive integer constant and then set $\gamma$ to be the smallest positive integer that satisfies the condition. The delay is bounded by the product of worst-case number of iterations between two answers of $\mA'$ and the work done between each iteration which is $(\delta / \gamma) \cdot (\gamma + \eta + c_{\textsf{pause}} + c_{\textsf{resume}}) \leq \delta \cdot (1 + T/T' + (c_{\textsf{pause}} + c_{\textsf{resume}})/\gamma ) = O(\delta \cdot \max\{1, T/T'\})$.
\end{proof}

Lemma~\ref{lem:helper:one} tells us that as long as $T = O(T')$, the output of $\mA$ and $\mA'$ can be combined without giving up on delay guarantees by pacing the output of $\mA'$. Note that we need to know the exact values of $T$ and $T'$. This can be accomplished by calculating the number of operations in the algorithms $\mA$ and $\mA'$ to bound their running time. The next lemma introduces our second key idea of interleaving stored output result with \emph{on-the-fly} query computation. Algorithm~\ref{algo:space:interleaving} describes the detailed algorithm for Lemma~\ref{lem:helper:two}.

\begin{lem} \label{lem:helper:two}
   Suppose an algorithm $\mA$ can enumerate all output tuples in $T$ time with no delay guarantees, where $T$ is known in advance. \reva{Suppose a data structure can enumerate $J$ output tuples without duplication with constant delay}. Then, an algorithm exists that enumerates all output tuples with delay guarantee $\delta = O(T/J)$.
\end{lem}
\begin{proof}
	Let $\delta$ be a parameter to be fixed later. \reva{Create two empty hash sets $H$ and $H'$ that will be used for deduplication}. Using a similar interleaving strategy as above, \reva{we use the data structure to obtain one result from $J$ in constant time, insert it into $H'$,} and allow algorithm $\mA$ to run for $\delta$ time. {Whenever $\mA$ wants to emit an output tuple, it probes the hash set $H$ and $H'$, emits $t$ only if $t$ does not appear in $H$ and $H'$, followed by inserting $t$ in $H$. Inserting $t$ in $H$ will ensure that $\mA$ does not output duplicates}\footnote{If $\mA$ guarantees that it will generate results with no duplicates, then there is no need to use $H$.}. Each probe and insertion takes $O(1)$ time, so the total running time of $\mA$ is $O(T)$. Our goal is to choose $\delta$ such that $\mA$ terminates before the materialized output $J$ runs out. This condition is satisfied when $\delta \cdot J \geq O(T)$ which gives us $\delta = O(T/J)$. It can be easily checked that no duplicated result is emitted and $O(\delta)$ delay is guaranteed between every pair of consecutive results. Again, observe that we need the algorithm $\mA$ to be pausable, which means that we should be able to resume the execution from where we left off. This can be achieved by storing the contents of all registers in the memory and loading it when required to resume execution.
\end{proof}

\begin{algorithm}[t]
	\SetCommentSty{textsf}
	\DontPrintSemicolon 
	\SetKwInOut{Input}{Input}
	\SetKwInOut{Output}{Output}
	\SetKwFunction{len}{\textsf{len()}}
	\Input{Materialized output list $J$, Algorithm $\mA$ with known completion time $T$}
	\Output{Deduplicated result of $\mA$}
	\SetKwData{ptr}{\textsf{ptr}}
	\SetKwData{dedup}{\textsf{dedup}}
	\SetKwData{counter}{\textsf{counter}}
	$\delta \leftarrow O(T/J),  \ptr \leftarrow 0, \dedup \leftarrow 0$, $H \leftarrow \emptyset$, $H' \leftarrow \emptyset$; \tcc*{empty hash-set}
	{\While{$\ptr < |J|$}{
			output $J[\ptr]$; \tcc*{output result from $J$ to maintain delay guarantee}
            insert $J[\ptr]$ in $H'$\;
			$\ptr \leftarrow \ptr + 1, \counter \leftarrow 0$; \;
			\While{$\counter \leq  \delta$}{
				\If{$\mA$ has not completed}{
					\label{line:pause} Execute $\mA$ for $c$ time; \tcc*{$c$ is a constant }
                        \ForEach{tuple $t$ generated (if any)}{
						\If{$t \not \in H'$ and $t \not \in H$}{
							output $t$ and insert $t$ in $H$; \;
						}
					}
				}
				$\counter \leftarrow \counter + c$; \;
			}	
	}}
	\caption{{\sc Deduplicate}$(J, \mA)$}
	\label{algo:space:interleaving}
\end{algorithm} 

The final helper lemma allows us to enumerate the union of (possibly overlapping) results of $m$ different algorithms where each algorithm outputs its result {according to a total order $\preceq$}, such that the union is also enumerated according to $\preceq$. {This lemma is based on the idea presented as Fact $3.1.4$ in~\cite{kazana2013query}.}

\begin{lem} \label{lem:helper:three}
	Consider $m$ algorithms $\mA_1, \mA_2, \cdots, \mA_m$ such that each $\mA_i$ enumerates its output $L_i$ with delay $O(\delta)$ according to the total order $\preceq$. Then, the union of their output can be enumerated (without duplicates) with $O(m \cdot \delta)$ delay according to $\preceq$. 
\end{lem}
\begin{proof}
	See~\autoref{algo:merging}. For simplicity of exposition, we assume that $\mA_i$ outputs a null value when it finishes enumeration. Note that results enumerated by one algorithm are in order. Thus, it always outputs the local minimum result ($e_i$) over the remaining result to be enumerated. \autoref{algo:merging} goes over all local minimum results over all algorithms and outputs the smallest one (denoted $w$) as the global minimum result (\autoref{line:min}). Once a result is enumerated, each $\mA_i$ must check whether its $e_i$ matches $w$. If yes, then $\mA_i$ needs to update its local minimum result by finding the next one. Then,~\autoref{algo:merging} repeats this loop until all algorithms finish enumeration. One distinct result is enumerated in each iteration of the while loop. It takes $O(m)$ time to find the globally minimum result and $O(m \cdot \delta)$ to update all local minimum results (\autoref{line:for:start}-\autoref{line:for:end}). Thus, Algorithm~\ref{algo:merging} has a delay of $O(m \cdot \delta)$.
\end{proof}

\begin{algorithm}[t]
	\SetCommentSty{textsf}
	\DontPrintSemicolon 
	\SetKwFunction{fullreducer}{\textsc{Materialize}}
	\SetKwFunction{initializepq}{\textsc{InitializePQ}}
	\SetKwFunction{recurse}{\textsc{Recurse}}
	\SetKwFunction{proc}{\textsf{eval}}
	\SetKwFunction{pop}{\textsc{pop}}
	\SetKwData{pointer}{\textsf{pointer}}
	\SetKwData{temp}{\textsf{temp}}
	\SetKwFunction{fillout}{\textsc{FillOUT}}
	\SetKwInOut{Input}{\textsc{input}}
	\SetKwInOut{Global}{\textsc{global variables}}
	\SetKwInOut{Output}{\textsc{output}}
	\SetKwProg{myproc}{\textsc{procedure}}{}{}
	
	$S \gets \{1,2,\cdots, m\}$; \label{line:one}\;
	\lForEach{$i \in S$}{$e_i \gets \mA_i.first()$; \label{line:three}}
	\While{$S \neq \emptyset$}{
		$w \gets \min_{i \in S} e_i$\label{line:min}; \tcc*{finds the smallest output (using $\preceq$) over all algorithms} 
		{\bf output} $w$\label{line:output}; \;
		\ForEach{$i \in S$\label{line:for:start}}{
			\lIf{$e_i = w$}{
				$e_i \gets \mA_i.next()$; \label{line:for:end}}
			\lIf{$e_i = null$}{$S \gets S- \{i\}$; 
   }
		}
	}
	\caption{{\sc Merge}$(\mA_1, \mA_2, \cdots, \mA_m)$}
	\label{algo:merging}
\end{algorithm}

Implied by Lemma~\ref{lem:helper:three}, the {\em list merge} problem can be done with delay guarantees: Given $m$ lists $L_1, L_2, \cdots, L_m$ whose elements are drawn from a common domain, if elements in $L_i$ are distinct (i.e. no duplicates) and {ordered according to $\preceq$}, then the union of all lists $\bigcup_{i=1}^m L_i$ can be enumerated {in sorted order given by $\preceq$} with delay $O(m)$. Note that the enumeration algorithm $\mA_i$ degenerates to going over elements one by one in list $L_i$, which has $O(1)$ delay guarantee as long as indexes/pointers within $L_i$ are well-built. 
Throughout the paper, we use this primitive as {\sc ListMerge}$(L_1, L_2, \cdots, L_m)$.
	\section{Star Queries}
\label{sec:proof}

We now study enumeration algorithms for star queries. 
We start with a detailed discussion on how our result improves over prior work in Section~\ref{sec:prior:Work}, then a warm-up proof for two-path query 
in Section~\ref{sec:warmup}, and the complete proof for general star queries in Section~\ref{sec:main}.
\begin{thm} \label{thm:star:delay}
	For a star query $\qstar{k}$ and any instance $D$, there is an algorithm with preprocessing time $T_p = O(|D|)$ and preprocessing space $S_p = O(|D|)$, such that we can enumerate $\qstar{k}(D)$ with
	$\delta = O\bigg(\frac{|D|^{k/k-1}}{|\tOUT_{\Join}|^{1/k-1}}\bigg)$ delay and $S_e = O(|D|)$ space.
\end{thm}
In the above theorem, the delay depends on the full join result size $|\tOUT_{\Join}| = |\qstar{k}(D)|$. As the join size increases, the algorithm can obtain better delay guarantees. In the extreme case when $|\tOUT_{\Join}| = \Theta(|D|^k)$, it achieves constant delay with linear time preprocessing. In the other extreme, when $|\tOUT_{\Join}| = \Theta(|D|)$, it achieves linear delay.  When $|\tOUT_{\Join}|$ has linear size, we can compute and materialize the query result in linear preprocessing time and achieve constant-delay enumeration. Generalizing this observation, we can always achieve constant delay when $T_p$ is sufficient to evaluate the full join result.

 \begin{figure*}[h]
	\begin{subfigure}{0.4\linewidth}
		\hspace*{-6em}
		\includegraphics[scale=0.50]{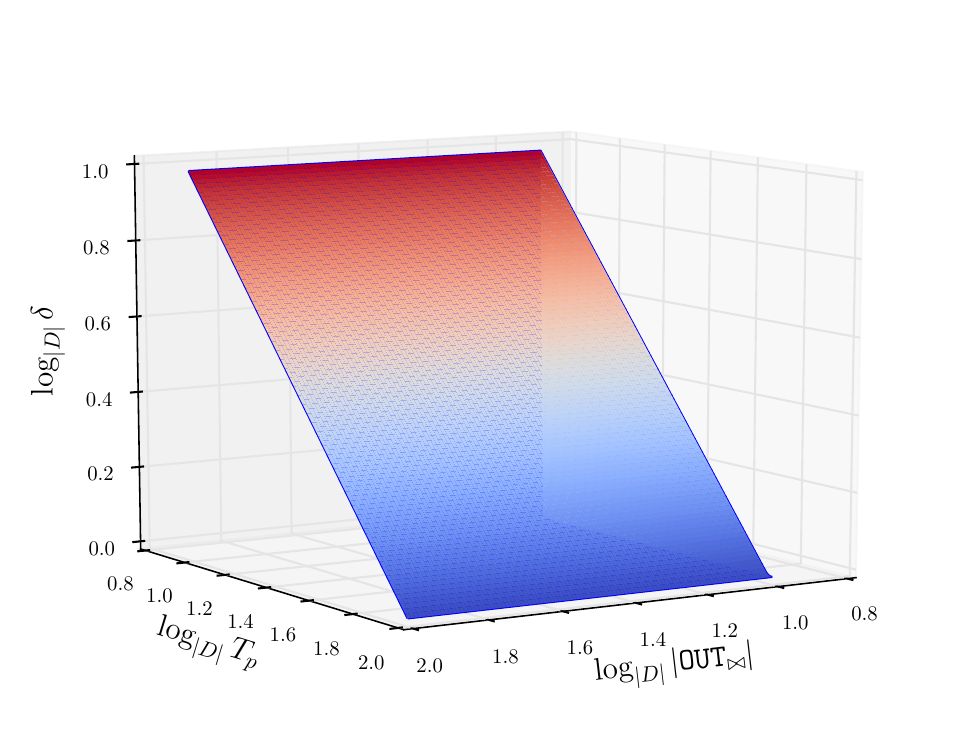}
		
		\label{subfig:one}
	\end{subfigure}
	\hspace*{-2em}
	\begin{subfigure}{0.4\linewidth}
		\includegraphics[scale=0.50]{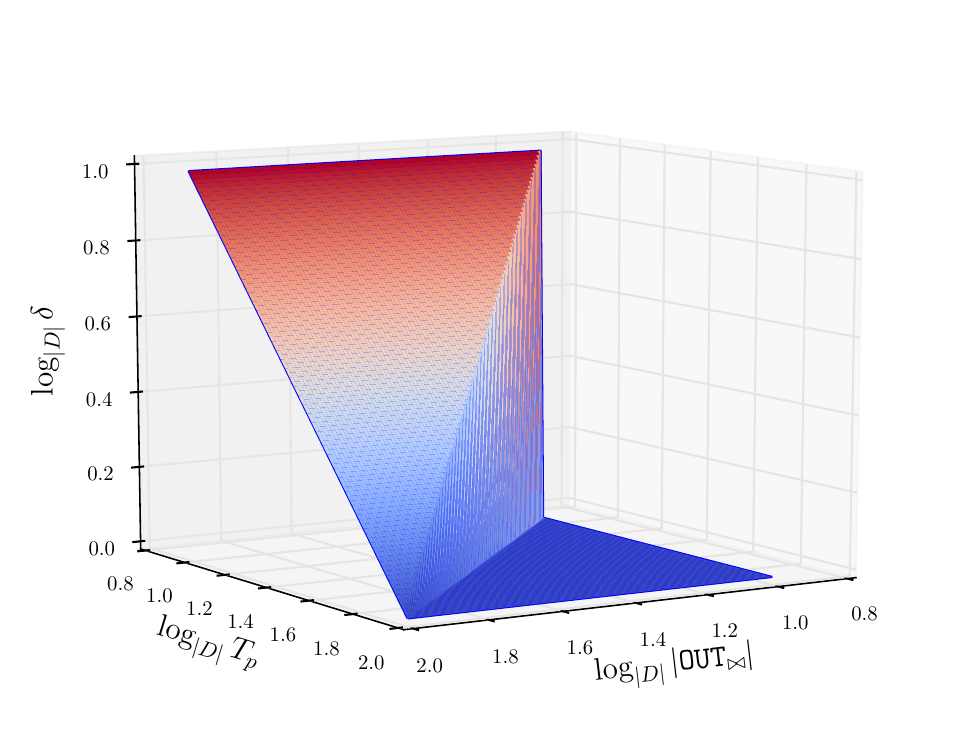}
		
		\label{subfig:two}		
	\end{subfigure}
    \vspace{-2em}
	\caption{Worst-case tradeoffs given by Theorem~\ref{thm:basic:three} without (left) and with (right) taking $|\tOUT_{\Join}|$ into consideration.}
	\label{fig:plots:one}
\end{figure*}
\begin{figure}[h]
 \vspace{-1.5em}
	\hspace*{-2em}
	\minipage{0.47\textwidth}
	\vspace{-1em}
	\includegraphics[scale=0.5]{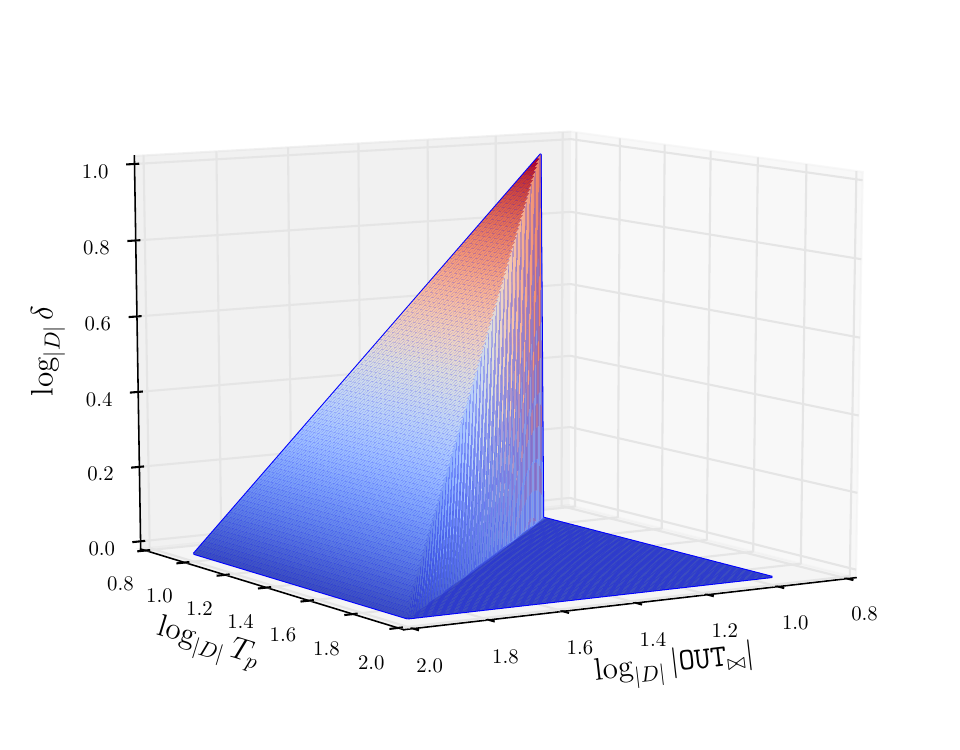}
    \vspace{-2em}
	\caption{Theorem~\ref{thm:star:delay} for $k=2$.} \label{subfig:four}	
	\endminipage \ \
	\minipage{0.52\textwidth}
	\vspace{1.75em}
	\begin{tikzpicture}
        \draw [<->,thick] (0,3.9) node (yaxis) [above] {}
        |- (5.7,0) node (xaxis) [below] {};
     \node [rotate=90] at (-0.6,2.0) {\small Delay $\log_{|D|} \delta$};   
     \node at (2.75,-0.6) {\small Preprocessing time $\log_{|D|} T_p$};
     \node[diamond,draw,scale=0.4,fill=green] at (3,0.3) { };
     \node[regular polygon,regular polygon sides=3,draw,scale=0.3,fill=blue] at (1,0.3) { };
     \node[circle,draw,scale=0.4,fill=red] at (1,2) { };
     \node[star,star points=5,draw,scale=0.3,fill=yellow] at (1,3.5) { };
     \draw[very thick,blue] (1,3.5) -- (3,2);
     \draw[dotted, black] (3,2) -- (3,0);
     \draw[dotted, black] (3,2) -- (0,2);
     \draw[dotted, black] (1,3.5) -- (1,0) node[below] {\tiny $1$};
     \draw[dotted, black] (3,0.3) -- (3,0) node[below] {\tiny $\log_{|D|} T_p$};
     \draw[dotted, black] (1,2) -- (0,2) node[left] {\tiny $\epsilon^\star$};
     \draw[dotted, black] (1,3.5) -- (0,3.5)  node[left] {\tiny $1$};
     \draw[dotted, black] (3,0.3) -- (0,0.3) node[left] {\tiny $0$};
     \matrix [draw,below left] at (current bounding box.north east) {
      \node [star,star points=5,draw,scale=0.3,fill=yellow,label=right:{\tiny $\alpha$-acyclic}] {}; \\
      \node [diamond,draw,scale=0.4,fill=green,label=right:{\tiny Conjunctive}] {}; \\
      \node [regular polygon,regular polygon sides=3,draw,scale=0.3,fill=blue,label=right:{\tiny free-connex}] {}; \\
      \node [circle,draw,scale=0.4,fill=red,label=right:{\tiny Theorem~\ref{thm:star:delay}}] {}; \\
    };
    \end{tikzpicture}
     \vspace{-1em}
	\caption{Trade-off in the worst-case for star query. Solid blue line shows Theorem~\ref{thm:basic:three}.} 
	\label{subfig:three}
	\endminipage\hfill
\end{figure}

\subsection{Comparison with Prior Work} \label{sec:prior:Work}

It is instructive now to compare the \changes{worst-case} delay guarantee obtained by Theorem~\ref{thm:basic:three} for $\qstar{k}(D)$ with Theorem~\ref{thm:star:delay}. {Suppose that we want to achieve delay $\delta = O(|D|^{1-\epsilon})$ for some $\epsilon \in [0, (\log_{|D|} |\tOUT_{\Join}| -1) / (k-1)]$. Theorem~\ref{thm:basic:three} tells us that this requires $O(|D|^{1+\epsilon(k-1)})$ preprocessing time. Then, it holds that:
	
$$ |D|^{1-\epsilon} \geq |D|^{1 - \frac{(\log_{|D|} |\tOUT_{\Join}| -1)}{k-1}} = |D|^{\frac{k - \log_{|D|}|\tOUT_{\Join}|}{k-1}} = |D|^{k/k-1} / |\tOUT_{\Join}|^{1/k-1}  $$

In other words, either we have enough preprocessing time to materialize the output and achieve constant delay, or we can achieve the desirable delay with linear preprocessing time.

Figure~\ref{fig:plots:one}, ~\ref{subfig:four} and  ~\ref{subfig:three} show the existing and new tradeoff results. Figure~\ref{fig:plots:one} shows the tradeoff curve obtained from Theorem~\ref{thm:basic:three} by adding $|\tOUT_{\Join}|$ as a third dimension, and adding the optimization for constant delay when $T_p \geq O(|\tOUT_{\Join}|)$. Figure~\ref{subfig:four} shows the tradeoff obtained from our result, while Figure~\ref{subfig:three} shows other existing results for a fixed value of $|\tOUT_{\Join}|$. For a fixed value of $|\tOUT_{\Join}|$, the delay guarantee does not change in Figure~\ref{subfig:four} as we increase $T_p$ from $|D|$ to $|\tOUT_{\Join}|$. It remains an open question to decrease the delay further if we allow more preprocessing time. Such an algorithm would correspond to a curve connecting the red point({\color{red} $\bullet$}) and the green diamond(\tikz{\node[draw=green,fill=green,diamond,shape border rotate=90,minimum
width=0.2cm,minimum height=0.2cm,inner sep=0pt] at (0,0) {};}) in Figure~\ref{subfig:three}. Our results thus imply that, depending on $|\tOUT_{\Join}|$, one must choose a different algorithm to achieve the optimal tradeoff between preprocessing time and delay. Since $|\tOUT_{\Join}|$ can be computed in linear time (using a simple adaptation of Yannakakis algorithm~\cite{yannakakis1981algorithms, pagh2006scalable}), this can be done without affecting the preprocessing bounds}.

\begin{figure}
	\begin{subfigure}{0.45\linewidth}
 \centering
		\scalebox{0.85}{
			\begin{tikzpicture}
			\tikzset{vertex/.style={circle, fill=black!25, minimum size=8pt}}
			
			\node[color=red] at (-5, 1) {$R(x,y)$};
			\node[color=blue] at (-3, 1) {$S(y,z)$};
			
			\node[vertex] (X-1) at (-6,-1) {$d_2$};
			\node[vertex] (Y-1) at (-4,-1) {$e_1$};
			\node[vertex] (Z-1) at (-2,-1) {$f_2$};
			\node[vertex] (X-2) at (-6,0) {$d_1$};		
			\node[vertex] (Z-2) at (-2,0) {$f_1$};
			\node[draw=none] at (-2,-1.5) {$\boldsymbol{\vdots}$};
			\node[draw=none] at (-6,-1.5) {$\boldsymbol{\vdots}$};				
			\node[vertex] (X-n) at (-6,-2.5) {$d_{{N}}$};
			\node[vertex] (Z-n) at (-2,-2.5) {$f_{{N}}$};

			\draw[color=red, line width=0.2mm] (X-1) -- node [above, color=black] {} (Y-1);
			\draw[color=red, line width=0.2mm] (X-2) -- node [above, color=black] {} (Y-1);		
			\draw[color=red, line width=0.2mm] (X-n) -- node [above, color=black] {} (Y-1);		
			\draw[color=blue, line width=0.2mm] (Y-1) -- node [above, color=black] {}  (Z-1);
			\draw[color=blue, line width=0.2mm] (Y-1) -- node [above, color=black] {}  (Z-2);
			\draw[color=blue, line width=0.2mm] (Y-1) -- node [above, color=black] {}  (Z-n);		
			\end{tikzpicture}}
		\caption{Database $D_0$ with full join size $N^{2}$.}	\label{fig:quadratic}
	\end{subfigure}
	\hfill
	\begin{subfigure}{0.45\linewidth}
    \centering
		\scalebox{0.85}{
			\begin{tikzpicture}
			\tikzset{vertex/.style={circle, fill=black!25, minimum size=8pt}}
			
			\node[color=red] at (-5, 1) {$R(x,y)$};
			\node[color=blue] at (-3, 1) {$S(y,z)$};
			
			\node[vertex] (X-1) at (-6,-1) {$a_2$};
			\node[vertex] (Y-1) at (-4,-1) {$b_2$};
			\node[vertex] (Z-1) at (-2,-1) {$c_2$};
			\node[vertex] (X-2) at (-6,0) {$a_1$};
			\node[vertex] (Y-2) at (-4,0) {$b_1$};			
			\node[vertex] (Z-2) at (-2,0) {$c_1$};
			\node[draw=none] at (-4,-1.5) {$\boldsymbol{\vdots}$};
			\node[draw=none] at (-2,-1.5) {$\boldsymbol{\vdots}$};
			\node[draw=none] at (-6,-1.5) {$\boldsymbol{\vdots}$};				
			\node[vertex] (X-n) at (-6,-2.5) {$a_{\sqrt{N}}$};
			\node[vertex] (Y-n) at (-4,-2.5) {$b_{\sqrt{N}}$};		
			\node[vertex] (Z-n) at (-2,-2.5) {$c_{\sqrt{N}}$};

			\draw[color=red, line width=0.2mm] (X-1) -- node [above, color=black] {} (Y-1);
			\draw[color=red, line width=0.2mm] (X-1) -- node [above, color=black] {} (Y-2);
			\draw[color=red, line width=0.2mm] (X-1) -- node [above, color=black] {} (Y-n);		
			\draw[color=red, line width=0.2mm] (X-2) -- node [above, color=black] {} (Y-1);		
			\draw[color=red, line width=0.2mm] (X-2) -- node [above, color=black] {} (Y-2);
			\draw[color=red, line width=0.2mm] (X-2) -- node [above, color=black] {} (Y-n);		
			\draw[color=red, line width=0.2mm] (X-n) -- node [above, color=black] {} (Y-1);		
			\draw[color=red, line width=0.2mm] (X-n) -- node [above, color=black] {} (Y-2);
			\draw[color=red, line width=0.2mm] (X-n) -- node [above, color=black] {} (Y-n);		
			\draw[color=blue, line width=0.2mm] (Y-1) -- node [above, color=black] {}  (Z-1);
			\draw[color=blue, line width=0.2mm] (Y-1) -- node [above, color=black] {}  (Z-2);
			\draw[color=blue, line width=0.2mm] (Y-1) -- node [above, color=black] {}  (Z-n);				
			\draw[color=blue, line width=0.2mm] (Y-2) -- node [above, color=black] {}  (Z-1);		
			\draw[color=blue, line width=0.2mm] (Y-2) -- node [above, color=black] {}  (Z-2);
			\draw[color=blue, line width=0.2mm] (Y-2) -- node [above, color=black] {}  (Z-n);		
			\draw[color=blue, line width=0.2mm] (Y-n) -- node [above, color=black] {}  (Z-1);	
			\draw[color=blue, line width=0.2mm] (Y-n) -- node [above, color=black] {}  (Z-2);			
			\draw[color=blue, line width=0.2mm] (Y-n) -- node [above, color=black] {}  (Z-n);
			\end{tikzpicture}}
		\caption{Database $D_1$ with full join size $N^{3/2}$.}	\label{fig:mmul}
	\end{subfigure}
	\caption{$D_0 \cup D_1$ forms a database where Theorem~\ref{thm:star:delay} improves the delay guarantee of Theorem~\ref{thm:basic:three}.}
\end{figure}

\changes{Next, we show how our result provides an algorithmic improvement over Theorem~\ref{thm:basic:three}. Consider the instances $D_0, D_1$ depicted in Figure~\ref{fig:quadratic} and Figure~\ref{fig:mmul} respectively, and assume we want to use linear preprocessing time \reva{(i.e., Theorem~\ref{thm:basic:three} is used with $\epsilon = 0$)}. 

 For $D_1$, the algorithm of Theorem~\ref{thm:basic:three} materializes nothing, since no $y$ valuation has a degree of $O(|D|^0)$, and the delay will be $\Theta(\sqrt{N})$. No materialization also occurs for $D_0$, but the delay will be $O(1)$ here. It is easy to check that our algorithm matches the delay in both instances. 
 Now, consider the instance $D = D_0 \cup D_1$. The input size for $D$ is $\Theta(N)$, while the full join size is $N^{3/2} + N^2 = \Theta(N^2)$. The algorithm of Theorem~\ref{thm:basic:three} will again achieve only a $\Theta(\sqrt{N})$ delay since after the linear time preprocessing, no $y$ valuations can be materialized. In contrast, our algorithm still guarantees a constant delay. This algorithmic improvement results from the careful overlapping of the constant-delay computation, for instance, $D_0$, with the computation for $D_1$. 
 
The above construction can be generalized as follows. Let $\alpha \in (0,1)$ be some constant. $D_0$ remains the same. For $D_1$, we construct $R$ to be the cross product of $N^\alpha$ $x$-values and $N^{1-\alpha}$ $y$-values, and $S$ to be the cross product of $N^\alpha$ $z$-values and $N^{1-\alpha}$ $y$-values. As before, let $D = D_0 \cup D_1$. The input size for $D$ is $\Theta(N)$, while the full join size is $N^{2-\alpha} + N^2 = \Theta(N^2)$. Hence, our algorithm achieves constant delay with linear preprocessing time. In contrast, the algorithm of Theorem~\ref{thm:basic:three} achieves $\Theta(N^{1-\alpha})$ delay with linear preprocessing time. In fact, the $\Theta(N^{1-\alpha})$ delay occurs even if we allow $O(N^{1+\epsilon})$ preprocessing time for any $\epsilon < \alpha$. We can similarly show that there also exists an instance where achieving constant delay using Theorem~\ref{thm:basic:three} requires near quadratic preprocessing time as shown Example~\ref{ex:quadratic} below.
}

\begin{exa} \label{ex:quadratic}
	We construct an instance where achieving constant delay with Theorem~\ref{thm:basic:three} would require close to $\Theta(|D|^2)$ computation. Let us fix $N$ to be a power of $2$. We will fix $|\domain(x)| = |\domain(z)| = N \log N$. Let $D_i$ be the database constructed by setting $N^\alpha = 2^i$ for $i \in \{1, 2, \dots, \log N\}$ where relation $R$ is the cross product of $N^\alpha$ $x$-values and $N^{1-\alpha}$ $y$-values, and $S$ is the cross product of $N^\alpha$ $z$-values and $N^{1-\alpha}$ $y$-values. We also construct a database $D^*$ consisting of a single $y$ connected to all $x$ and $z$ values. Let $D = D^* \cup D_1 \cup D_2 \cup \dots \cup D_{\log N}$. It is now easy to see that $|D| = N \cdot \log N, |\domain(y)| = \sum_{\alpha} N^{1- \alpha} \leq 2N = \Theta(|D|/\log |D|)$ and $|\tOUT_{\Join}| = \sum_{\alpha} N^{1+\alpha
	} + N^2 \log ^2 N = \Theta(|D|^2)$. On this instance, Theorem~\ref{thm:basic:three} achieves $\Theta(|D|/\log |D|)$ after linear time preprocessing. Suppose we target constant-delay enumeration. Let us fix this constant as $c^*$ (a power of $2$) for simplicity. Then, we need enough preprocessing time to materialize the join result of all database instances $D_i$ where $i \in \{1, 2, \dots, \log (N/c^*)\}$ to ensure that the number of heavy $y$ values that remain is at most $c^*$. Hence, $T_p > \sum_{i \in \{1, 2, \dots, \log (N/c^*)\}} N \cdot 2^i > N^2/c^* = \Theta(|D|^2/ \log |D|)$. Theorem~\ref{thm:basic:three} requires near quadratic computation to achieve constant-delay enumeration. 
\end{exa}

\subsection{Warm-up: Two-Path Query}
\label{sec:warmup}
As warm-up, we show an algorithm for $Q_{\texttt{\upshape two-path}} = \pi_{x,z} (R(x,y) {\Join} S(y,z))$ that achieves $O(|D|^2/|\tOUT_{\Join}|)$ delay with linear preprocessing time. At a high level, we will decompose the join into two subqueries with disjoint outputs. The subqueries will be generated based on whether a valuation for $x$ is \emph{light} or not based on its degree in relation $R$. For all light valuations of $x$ (degree at most $\delta$), we will show that their enumeration is achievable with delay $\delta$. For the heavy $x$ valuations, we will show that they also can be computed {\em on-the-fly} while maintaining the delay guarantees.

\introparagraph{Preprocessing Phase} We first process the input relations by removing dangling\footnote{\revb{A tuple $t \in R_\ell(\bx_\ell)$ is said to be dangling if $R_1(\bx_1) \Join \dots \Join R_{\ell-1}(\bx_{\ell-1}) \Join t \Join R_{\ell+1}(\bx_{\ell+1}) \Join \dots R_n(\bx_n)$ is empty.}} tuples. During the preprocessing phase, we will store the input relations as a hash map and sort the valuations for $x$ in increasing order of their degree. {Using any comparison based sorting technique requires $\Omega(|D| \log |D|)$ time in general. Thus, if we wish to remove the $\log |D|$ factor, we must use non-comparison-based sorting algorithms. In this paper, we will use count sort~\cite{cormen2009introduction}, which has complexity $O(|D| + r)$ where $r$ is the range of the non-negative key values. However, we need to ensure that all relations in the database $D$ satisfy the bounded range requirement. This can be easily accomplished by introducing a bijective function $f: \domain(D) \rightarrow \{1, 2, \dots, |D|\}$ that maps all values in the active domain of the database to some integer between $1$ and $|D|$ (both inclusive). Both $f$ and its inverse $f^{-1}$ can be stored as hash tables: suppose there is a counter $c \leftarrow 1$. We perform a linear pass over the database and check if some value $v \in \domain(D)$ has been mapped or not (by checking if there exists an entry $f(v)$). If not, we set $f(v) = c, f^{-1}(c) = v$ and increment $c$. Once the hash tables $f$ and $f^{-1}$ have been created, we modify the input relation $R$ (and $S$ similarly) by replacing every tuple $t \in R$ with tuple $t' = f(t)$. Since the mapping is a relabeling scheme, such a transformation preserves the degree of all the values. The codomain of $f$ is also equipped with a total order $\preceq$ (we will use $\leq$). Note that $f$ is not an order-preserving transformation in general, but our algorithms do not require this property.} 

Next, for every tuple $t \in R(x,y)$, we create a hash map with key $\pi_x (t)$, and the value is a list to which $\pi_y (t)$ is appended; and for every tuple $t \in S(y,z)$, we create a hash map with key $\pi_y (t)$ and the value is a list to which $\pi_z (t)$ is appended. For the second hash map, we sort the value list using {sort order $\preceq$} for each key once each tuple $t \in S(y,z)$ has been processed. Finally, we sort all values in $\pi_{x}(R)$ in increasing order of their degree in $R$ (i.e., $|\sigma_{x = v_i} R(x,y)|$ is the sort key). Let $\mL = \{v_1, \dots, v_n\}$ denote the ordered set of these values sorted by their degree and let $d_1, \dots, d_n$ be their respective degrees. {Creating the sorted list $\mL$ takes $O(|D|)$ time since the degrees $d_i$ satisfy the bounded range requirement (i.e $1 \leq d_i \leq |D|$). Next, we identify the smallest index $i^*$ such that 
\begin{align}
\sum_{v : \{v_1, v_2, \dots, v_{i^*}\}} |R(v, y) \Join S(y,z)| \geq \sum_{v : \{ v_{i^*+1}, \dots, v_n\}} |R(v, y) \Join S(y,z)| \label{cond}
\end{align} 
This can be computed by doing a linear pass on $\mL$ using a simple adaptation of Yannakakis algorithm~\cite{yannakakis1981algorithms, pagh2006scalable}.} This entire phase takes time $O(|D|)$.

\introparagraph{Enumeration Phase} The enumeration algorithm interleaves the following two loops using the construction in Lemma~\ref{lem:helper:one}. Specifically, it will spend an equal amount of time (a constant) before switching to the computation of the other loop.
\begin{algorithm}[t]
	\SetCommentSty{textsf}
	\DontPrintSemicolon 
	\SetKwFunction{beginn}{\textsf{begin()}}
	\SetKwFunction{endd}{\textsf{end()}}
	\SetKwData{left}{\textsf{left}}
	\SetKwData{right}{\textsf{right}}
	\SetKwFunction{fillout}{\textsc{FillOUT}}
	\SetKwInOut{Input}{\textsc{input}}
	\SetKwInOut{Global}{\textsc{global variables}}
	\SetKwInOut{Output}{\textsc{output}}
	\SetKwProg{myproc}{\textsc{procedure}}{}{}
	\BlankLine
	\For{$i=1, \dots,  i^*$}{
		Let $\pi_y (\sigma_{x = v_{i}}(R)) = \{u_1, u_2, \cdots, u_\ell\}$;\;
		output $(v_i, f^{-1}(${\sc ListMerge}$(\pi_z \sigma_{y = u_1} S,\pi_z \sigma_{y = u_2} S, \cdots, \pi_z \sigma_{y = u_\ell} S)))$\footnotemark 
		}	
	\begin{tikzpicture}[remember picture, overlay]
	\draw [black, thick] (12,0.4) to [square right brace] (12,1.7);
	\node at (13.5, 1.1) {\tiny run for $O(1)$ time};
	\node at (13.5, 0.9) {\tiny then switch};
	\draw [black, thick] (12,0.35) to [square left brace] (12,-1.);
	\node at (13.5, -0.4) {\tiny run for $O(1)$ time};
	\node at (13.5, -0.6) {\tiny then switch};	
	\end{tikzpicture}
	\For{$i=i^*+1, \dots,  n$}{
		Let $\pi_y (\sigma_{x = v_{i}}(R)) = \{u_1, u_2, \cdots, u_\ell\}$;\;
		output $(v_i, f^{-1}(${\sc ListMerge}$(\pi_z \sigma_{y = u_1} S,\pi_z \sigma_{y = u_2} S, \cdots, \pi_z \sigma_{y = u_\ell} S)))$ 
		}	
	\caption{{\sc EnumTwoPath}} \label{algo:enumeration:twopath}
\end{algorithm}\footnotetext{Abusing the notation, $f^{-1}(\mB)$ for some ordered list (or tuple) $\mB$ returns an ordered list (tuple) $\mB'$ where $\mB'(i) = f^{-1}(\mB(i))$.}    
 The algorithm alternates between low-degree and high-degree values in $\mL$. The main idea is that, for a given $v_i \in \mL$, we can enumerate the result of the subquery $\sigma_{x=v_i}(Q_{\texttt{\upshape two-path}})$ with delay $O(d_i)$. This can be accomplished by Algorithm~\ref{algo:merging}, since the subquery is equivalent to list merging.
 
 \begin{exa} \label{ex:enum:two}
 	\begin{figure*}
 		\centering
 		\begin{subfigure}{0.13\linewidth}
 			\vspace{12pt}
 			\begin{tabular}{ !{\vrule width1pt} c| c|c|c  !{\vrule width1pt} } 
 				\Xhline{1pt}
 				$x$ & $y$ \\ 
 				\Xhline{1pt}
 				$a_1$ & $b_1$ \\ 
 				$a_1$ & $b_2$ \\
 				$a_1$ & $b_3$ \\
 				$a_2$ & $b_1$ \\
 				\Xhline{1pt}
 			\end{tabular}		
 			
 			\caption{Table $R$} \label{subfig:RelationR}	
 		\end{subfigure}
 		\begin{subfigure}{0.13\linewidth}
 			\vspace{12pt}
 			\begin{tabular}{ !{\vrule width1pt} c| c|c|c !{\vrule width1pt} } 
 				\Xhline{1pt}
 				$y$ & $z$ \\ 
 				\Xhline{1pt}
 				$b_1$ & $c_1$ \\ 
 				$b_1$ & $c_2$ \\
 				$b_2$ & $c_2$ \\
 				$b_3$ & $c_3$ \\
 				\Xhline{1pt}
 			\end{tabular}		
 			\caption{Table $S$} \label{subfig:RelationS}	
 		\end{subfigure}
 		\begin{subfigure}{0.28\linewidth}
 			\begin{tikzpicture}
 			\def\y{0}	
 			\node (a2) at (0,-1+\y) {\color{blue}$a_2$};
 			\node (b1) at (1,-1+\y)  {$b_1$};
 			
 			\draw[->, color=black] (a2.east) to[out=0,in=-180] (b1.west);
 			
 			\node (a1) at (0,-2+\y) {\color{red}$a_1$};
 			\node (b2) at (1,-2+\y)  {$b_2$};
 			\node (b3) at (1,-3+\y)  {$b_3$};
 			
 			\draw[->, color=black] (a1.east) to[out=0,in=-180] (b1.west);
 			\draw[->, color=black] (a1.east) to[out=0,in=-180] (b2.west);
 			\draw[->, color=black] (a1.east) to[out=0,in=-180] (b3.west);
 			
 			\node (c1) at (2.25,-1+\y) {$[c_1, c_2]$};
 			\node  at (3.25,-1+\y) {$S[b_1]$};
 			\node (c2) at (2,-2+\y) {$[c_2]$};
 			\node  at (3.25,-2+\y) {$S[b_2]$};
 			\node (c3) at (2,-3+\y) {$[c_3]$};
 			\node  at (3.25,-3+\y) {$S[b_3]$};
 			
 			\node at (1.9,-0.7+\y) {\color{blue}$\downarrow$};
 			\node at (2.1,-0.7+\y) {\color{red}$\downarrow$};
 			\node at (2,-1.7+\y) {\color{red}$\downarrow$};
 			\node at (2,-2.7+\y) {\color{red}$\downarrow$};
 			
 			\draw[->, color=black] (b1.east) to[out=0,in=-180] (c1.west);
 			\draw[->, color=black] (b2.east) to[out=0,in=-180] (c2.west);
 			\draw[->, color=black] (b3.east) to[out=0,in=-180] (c3.west);
 			
 			\end{tikzpicture}
 			\caption{output $(a_1, c_1)$} \label{subfig:preprocessing}
 		\end{subfigure}
 		\begin{subfigure}{0.20\linewidth}
 			\begin{tikzpicture}
 			\def\y{0}
 			\node (b1) at (1,-1+\y)  {$b_1$};
 			
 			\node (a1) at (0,-2+\y) {\color{red}$a_1$};
 			\node (b2) at (1,-2+\y)  {$b_2$};
 			\node (b3) at (1,-3+\y)  {$b_3$};
 			
 			\draw[->, color=black] (a1.east) to[out=0,in=-180] (b1.west);
 			\draw[->, color=black] (a1.east) to[out=0,in=-180] (b2.west);
 			\draw[->, color=black] (a1.east) to[out=0,in=-180] (b3.west);
 			
 			\node (c1) at (2.25,-1+\y) {$[c_1, c_2]$};
 			\node (c2) at (2,-2+\y) {$[c_2]$};
 			\node (c3) at (2,-3+\y) {$[c_3]$};
 			\node at (2.5,-0.7+\y) {\color{red}$\downarrow$};
 			\node at (2,-1.7+\y) {\color{red}$\downarrow$};
 			\node at (2,-2.7+\y) {\color{red}$\downarrow$};
 			
 			\draw[->, color=black] (b1.east) to[out=0,in=-180] (c1.west);
 			\draw[->, color=black] (b2.east) to[out=0,in=-180] (c2.west);
 			\draw[->, color=black] (b3.east) to[out=0,in=-180] (c3.west);
 			
 			\end{tikzpicture}
 			\caption{output $(a_1, c_2)$} \label{subfig:step1}
 		\end{subfigure}
 		\begin{subfigure}{0.20\linewidth}
 			\begin{tikzpicture}
 			\def\y{0}
 			\node (b1) at (1,-1+\y)  {$b_1$};
 			
 			\node (a1) at (0,-2+\y) {\color{red}$a_1$};
 			\node (b2) at (1,-2+\y)  {$b_2$};
 			\node (b3) at (1,-3+\y)  {$b_3$};
 			
 			\draw[->, color=black] (a1.east) to[out=0,in=-180] (b1.west);
 			\draw[->, color=black] (a1.east) to[out=0,in=-180] (b2.west);
 			\draw[->, color=black] (a1.east) to[out=0,in=-180] (b3.west);
 			
 			\node (c1) at (2.25,-1+\y) {$[c_1, c_2]$};
 			\node (c2) at (2,-2+\y) {$[c_2]$};
 			\node (c3) at (2,-3+\y) {$[c_3]$};
 			\node at (2.8,-0.7+\y) {\color{red}$\downarrow$};
 			\node at (2.3,-1.7+\y) {\color{red}$\downarrow$};
 			\node at (2,-2.7+\y) {\color{red}$\downarrow$};
 			
 			\draw[->, color=black] (b1.east) to[out=0,in=-180] (c1.west);
 			\draw[->, color=black] (b2.east) to[out=0,in=-180] (c2.west);
 			\draw[->, color=black] (b3.east) to[out=0,in=-180] (c3.west);
 			
 			\end{tikzpicture}
 			\caption{output$(a_1, c_3)$} \label{subfig:step2}
 		\end{subfigure}
 		\caption{Example for two path query enumeration}
 	\end{figure*}
 	
 	Consider relations $R$ and $S$ as shown in~\autoref{subfig:RelationR} and \autoref{subfig:RelationS}.~\autoref{subfig:preprocessing} shows the sorted valuations $a_2$ and $a_1$ by their degree and the valuations for $Z$ as sorted lists $S[b_1], S[b_2]$ and $S[b_3]$. For both $a_1$ and $a_2$, the pointers point to the head of the lists. We will now show how {\sc ListMerge}$(S[b_1], S[b_2], S[b_3])$ is executed for $a_1$. Since three sorted lists need to be merged, the algorithm finds the smallest valuation across the three lists. $c_1$ is the smallest valuation and the algorithm outputs $(a_1, c_1)$. Then, we need to increment pointers of all lists pointing to $c_1$ ($S[b_1]$ is the only list containing $c_1$).~\autoref{subfig:step1} shows the state of pointers after this step. The pointer for $S[b_1]$ points to $c_2$, and all other pointers still point to the list head. Next, we merge the list by finding each list's smallest valuation. Both $S[b_1]$ and $S[b_2]$ pointers are pointing to $c_2$ and the algorithm outputs $(a_1, c_2)$. The pointers for both $S[b_1]$ and $S[b_2]$ are incremented, and the enumeration for both lists is complete as shown in~\autoref{subfig:step2}. In the last step, only $S[b_3]$ list remains and we output $(a_1, c_3)$ and increment the pointer for $S[b_3]$. All pointers are now past the end of the lists 
 	and the enumeration is now complete.

 \end{exa}

\begin{thm} \label{lem:twopath:delay:1}
	For $Q_{\texttt{\upshape two-path}}$ and any instance $D$, there is an algorithm with preprocessing time $T_p = O(|D|)$ and preprocessing space $S_p = O(|D|)$, such that the query result $Q_{\texttt{\upshape two-path}}(D)$ can be enumerated with delay $\delta = O(|D|^2/|\tOUT_{\Join}|)$ and space $S_e = O(|D|)$.
\end{thm}
\begin{proof}
	{To prove this result, we will apply Lemma~\ref{lem:helper:one}, where $\mA'$ is the first loop (the one with light-degree values), and $\mA$ is the second loop (the one with high-degree values). 
		
		Let $\delta$ denote the degree of the valuation $v_{i^*}$. First, we claim that the delay of $\mA'$ will be $O(\delta)$. Indeed, {\sc ListMerge} will output a result every $O(\delta)$ time since the degree of each valuation in the first loop is at most $\delta$. Let $J_h = \sum_{i > i^*} |R(v_i,y) {\Join} S(y,z)|$ and $J_\ell = \sum_{i \leq i^*} |R(v_i,y) {\Join} S(y,z)|$. Then, $\mA'$ runs in time at least $J_\ell$, and $\mA$ in time at most $c^\star \cdot J_h$. Here, $c^\star$ is an upper bound on the number operations in each loop iteration in Algorithm~\ref{algo:enumeration:twopath}. Since by construction $J_\ell \geq J_h$, Lemma~\ref{lem:helper:one} obtains a total delay of $O(\delta)$.
		
		It now remains to bound $\delta$. First, as $i^*$ is the smallest index that satisfies~\autoref{cond}, it must be that  $J_\ell - J_h \leq |D|$ (if not, shifting the smallest index by one decreases the LHS by at most $|D|$ and increases the RHS by at most $|D|$ while still satisfying the condition that $J_\ell \geq J_h$). Combined with the fact that $J_\ell + J_h = |\tOUT_{\Join}|$, we get that $J_h \geq |\tOUT_{\Join}|/2 -  |D|/2  \geq 1/4 \cdot |\tOUT_{\Join}|$ assuming $|\tOUT_{\Join}| \geq 2 \cdot |D|$. Finally, $J_h \leq |D|^2/\delta$ since there are at most $|D|/\delta$ heavy values, and each heavy value can join with at most $|D|$ tuples for the full join. Combining the two inequalities gives us the claimed guarantee.}
\end{proof}

The reader should note that the delay of $\delta = O(|D|^2/|\tOUT_{\Join}|)$ is only an upper bound. Depending on the skew in the database instance, Algorithm~\ref{algo:enumeration:twopath} may achieve much better delay guarantees in practice, as shown in Example~\ref{ex:two} below.

\begin{exa} \label{ex:two}
Consider relation $R(x,y)$ of size $O(N)$, with values $v_1, \dots, v_N$ in $x$. Suppose that each of $v_1, \dots, v_{N-1}$ has a degree exactly 1, and each is connected to a distinct value in $y$. Also, $v_N$ has degree $N-1$ and is connected to all $N-1$ values in $y$. Let relation $S(z,y) = \{t \in R (x,y)\}$ (i.e., relation $S$ is a copy of relation $R$). Suppose we want to compute $Q_{\texttt{\upshape two-path}}$. It is easy to see that $\tOUT_{\Join}=\Theta(N)$. Thus, applying the bound of $\delta = O(N^2/|\tOUT_{\Join}|)$ gives us $O(N)$ delay. However, Algorithm~\ref{algo:enumeration:twopath} achieves a constant delay guarantee since all of $v_1, \dots, v_{N-1}$ are processed by the \textsf{left} pointer in $O(1)$ delay as they produce exactly one output result, while the \textsf{right} pointer processes $v_N$ on the fly in $O(N)$ time.
\end{exa}

\subsection{Proof of Theorem ~\ref{thm:star:delay}} 
\label{sec:main}

We are ready to prove Theorem~\ref{thm:star:delay} formally.
We now generalize Algorithm~\ref{algo:enumeration:twopath} for any star query.  At a high level, we will decompose the join query $\qstar{k}$ into a union of $k+1$ subqueries whose output is a partition of the result of the original query. These subqueries will be generated based on whether a value for some $x_i$ is \emph{light} or not. We will show that if any of the values for $x_i$ are light, the enumeration delay is small. The $(k+1)$-th subquery will contain heavy values for all attributes. Our key idea again is to {interleave} the join computation of the {\em heavy} subquery with the remaining light subqueries.

\introparagraph{Preprocessing Phase} Suppose the dangling tuples have been removed from all input relation, which can be achieved in linear time~\cite{yannakakis1981algorithms}. The full join size $|\tOUT_{\Join}|$ can also be computed in linear time. {Similar to the preprocessing phase in the previous section, we construct the hash tables $f, f^{-1}$ to perform the domain compression and modify all the input relations by replacing tuple $t$ with $f(t)$.} Set $\Delta = ({2 \cdot |D|^k}/{|\tOUT_{\Join}|})^{\frac{1}{k-1}}$. For each relation $R_i$, a value $v$ for attribute $x_i$ is {\em heavy} if its degree (i.e $|\pi_y \sigma_{x_i = v} R(x_i, y)|$) is greater than $\Delta$, and {\em light} otherwise. Moreover, a tuple $t \in R_i$ is identified as heavy or light depending on whether $\pi_{x_i}(t)$ is heavy or light. In this way, each relation $R$ is divided into two relations $R^h$ and $R^\ell$, containing heavy and light tuples, respectively, in time $O(|D|)$. The original query can be decomposed into subqueries of the following form:
$\pi_{x_1, x_2, \cdots, x_k} (R^?_1 \Join R^?_2 \Join \cdots \Join R^?_k)$, 
where $?$ can be either $h, \ell$ or $\star$. Here, $R^\star_i$ denotes the original relation $R_i$. However, care must be taken to generate the subqueries so that there is no overlap between the output of any subquery. In order to do so, we create $k$ subqueries of the form $$Q_i = \pi_{x_1, \dots, x_k} (R^h_1 \Join \cdots \Join R^h_{i-1} \Join R^\ell_i \Join R^\star_{i+1} \Join \cdots \Join R^\star_k)$$ 
In $Q_i$, relation $R_i$ has superscript $\ell$, all relations $R_1, \dots, R_{i-1}$ have superscript $h$ and relations $R_{i+1}, \dots, R_k$ have superscript $\star$. The $(k+1)$-th query with all $?$ as $h$ is denoted by $Q_H$. {Note that each output tuple $t$ is generated by exactly one of the $Q_i$; thus, the output of all subqueries is disjoint. This implies that each $f^{-1}(t)$ is also generated by exactly one subquery.} Similar to the preprocessing phase of the two path query, we store all $R^\ell_i$ and $R^h_i$ in hashmaps where the values in the maps are lists sorted in lexicographic order. 

\introparagraph{Enumeration Phase} We next describe how enumeration is performed. We will show that for $Q_L = Q_1 \cup \dots \cup Q_{k}$, the result can be enumerated with delay $O(\Delta)$.
As $Q_H$ contains all heavy valuations from all relations, we compute its join \emph{on-the-fly} by alternating between some subquery in $Q_L$ and $Q_H$. This will ensure that we can give some output to the user with delay guarantees and also make progress on computing the full join of $Q_H$. {Our goal is to reason about the running time of enumerating $Q_L$ (denoted by $T_L$) and the running time of $Q_H$ (denoted by $T_H$) and make sure that while we compute $Q_H$, we do not run out of the output from $Q_L$}. Next, we introduce the algorithm that enumerates output for any specific valuation $v$ of attribute $x_i$, described in Lemma~\ref{lem:star:degree}. This algorithm can be viewed as another instantiation of Algorithm~\ref{algo:merging}.

\begin{lem} \label{lem:star:degree}
	For a star query $\qstar{k}$, any instance  $D$, and an arbitrary value $v_i \in \domain(x_i)$ with degree $d$ in $R_i(x_i, y)$, its query result $\sigma_{x_i = v_i}\qstar{k}(D)$ can be enumerated with 
    $O(d)$ delay.
\end{lem}
\begin{proof}
	Consider some tuple $(v_i, u) \in R_i$. Each $u$ is associated with a list of valuations over attributes $(x_1, \cdots, x_{i-1}, x_{i+1}, \cdots, x_k)$, which is a Cartesian product of $k-1$ sub-lists $\sigma_{y = u} R_j(x_j, y)$. Note that such a list is not materialized as that for a two-path query but is present in a factorized form. We next define the enumeration algorithm $\mA_u$ for each $u \in \pi_y \sigma_{x_i = v_i} R_i(x_i,y)$, with lexicographical ordering of attributes $(x_1, \cdots, x_{i-1}, x_{i+1}, \cdots, x_k)$. Note that elements in each list $\pi_{x_j} \sigma_{y=u} R^?_j(x_j, y)$ can be enumerated with $O(1)$ delay. Then, $\mA_u$ enumerates all results in  $\times_{j \neq i: j \in \{1,2\cdots,k\}} \sigma_{y = u} R^?_j(x_j, y)$ by $k-1$ level of nested loops in lexicographic order $(x_1, \cdots, x_{i-1}, x_{i+1}, \cdots, x_k)$, which has $O(k-1) = O(1)$ delay. After applying Algorithm~\ref{algo:merging}, we can obtain an enumeration algorithm that enumerates the union of query results over all neighbors with $O(d)$ delay guarantee. {For each output tuple $t$ generated by Algorithm~\ref{algo:merging}, we return $f^{-1}(t)$ to the user. }
\end{proof}

\changes{Let $c^\star$ be an upper bound on the number of operations in each iteration of {\sc ListMerge}. This can be calculated by counting the number of operations in the exact implementation of the algorithm.
Directly implied by Lemma~\ref{lem:star:degree}, the result of any subquery in $Q_L$ can be enumerated with delay $ O(\Delta)$. Let $Q_H^*$ denote the corresponding full query of $Q_H$, i.e., the head of $Q_H^*$ also includes the variable $y$ ($Q_L^*$ is defined similarly). Then, $Q_H^*$ can be evaluated in time  $T_H \leq c^\star \cdot |Q_H^*| \leq c^\star \cdot |\tOUT_{\Join}|/2$ by using {\sc ListMerge} on subquery $Q_H$. This follows from the bound $|Q_H^*| \leq |D| \cdot (|D|/\Delta)^{k-1}$ and our choice of $\Delta = ({2 \cdot |D|^k}/{|\tOUT_{\Join}|})^{\frac{1}{k-1}}$. Since $|Q_H^*| + |Q_L^*| = |\tOUT_{\Join}|$, it holds that $|Q_L^*| \geq |\tOUT_{\Join}|/2$ given our choice of $\Delta$. Also, the running time $T_L$ is lower bounded by $|Q_L^*|$ (since we need at least one operation for every result). Thus, $T_L \geq |\tOUT_{\Join}|/2$. We are now ready to apply Lemma~\ref{lem:helper:one} with the following parameters: 
\begin{enumerate}
	\item $\mA$ is the full join computation of $Q_H$ and $T  = c^\star \cdot |\tOUT_{\Join}|/2$.
	\item $\mA'$ is the enumeration algorithm applied to $Q_L$ with delay guarantee $\delta = O(\Delta)$ and $T' = |\tOUT_{\Join}|/2$.	
	\item $T$ and $T'$ are fixed once $|\tOUT_{\Join}|, \Delta$, and the constant $c^\star$ are known.
\end{enumerate}

 By construction, the outputs of $Q_H$ and $Q_L$ are also disjoint. Thus, the conditions of Lemma~\ref{lem:helper:one} apply, and we obtain a delay of $O(\Delta)$. \hfill$\square$}

\subsection{Interleaving with Join Computation} \label{subsec:interleaving}
 Theorem~\ref{thm:star:delay} obtains poor delay guarantees when the full join size $|\tOUT_{\Join}|$ is close to input size $|D|$. In this section, we present an alternate algorithm that provides good delay guarantees.  The algorithm is an instantiation of Lemma~\ref{lem:helper:two} on the star query, which degenerates to computing as many distinct output results as possible in limited preprocessing time. An observation is that for each valuation $u$ of attribute $y$, the Cartesian product $\times_{i \in \{1,2,\cdots,k\}} \pi_{x_i} \sigma_{y = u} R_i(x_i, y)$ is a subset of output results without duplication. Thus, this subset of the output result is readily available since no deduplication needs to be performed. Similarly, after all relations are reduced, it is also guaranteed that each valuation of attribute $x_i$ of relation $R_i$ generates at least one output result. Thus, $\max_{i=1}^k |\domain(x_i)|$ results are also readily available that do not require deduplication. We define $J$ as the larger of the two quantities, i.e,  $J = \max \left\{\max_{i=1}^k |\domain(x_i)|,  \max_{u \in \domain(y)} \prod_{i =1}^k |\sigma_{y = u} R_i(x_i, y)| \right\}$. Together with these observations, we can achieve the following theorem.


\begin{thm} \label{thm:star:delay:alternate}
	For a star query $\qstar{k}$ and any instance $D$, there exists an algorithm with preprocessing time $O(|D|)$ and space $O(|D|)$, such that the query result $\qstar{k}(D)$ can be enumerated with
	$
	\text{delay } \delta = O\bigg({|\tOUT_{\Join}|}/{|\tOUT_{\pi}|^{1/k}}\bigg) \text{and space } S_e = O(|D|)
	$.
\end{thm}

In the above theorem, we obtain delay guarantees that depend on both the size of the full join result $\tOUT_{\Join}$ and the projection output size $\tOUT_{\pi}$. However, one does not need to know $|\tOUT_{\Join}|$ or $|\tOUT_{\pi}|$ to apply the result. We first compare the result with Theorem~\ref{thm:star:delay}. First, observe that both Theorem~\ref{thm:star:delay:alternate} and Theorem~\ref{thm:star:delay} require $O(|D|)$ preprocessing time. 
{Second, the delay guarantee provided by Theorem~\ref{thm:star:delay:alternate} can be better than Theorem~\ref{thm:star:delay}. This happens when $|\tOUT_{\Join}| \leq |D| \cdot J^{1-1/k}$, a condition that can be easily checked in linear time}.

We now proceed to describe the algorithm. First, we compute all the statistics for computing $J$ in linear time. If $J = |\domain(x_j)|$ for some integer $j \in \{1,2,\cdots,k\}$, we just materialize one result for each valuation of $x_j$. Otherwise, $J = \prod_{i =1}^k |\sigma_{y = u} R_i(x_i, y)|$ for some valuation $u$ in attribute $y$. Note that we do not need to materialize the Cartesian product explicitly but only need to store the tuples $\sigma_{y = u} R_i(x_i, y)$ for all $i \in \{1, \dots, k\}$. As mentioned before, each output in $\times_{i =1}^k \left(\pi_{x_i} \sigma_{y = u}R_i(x_i, y)\right)$ can be enumerated with $O(1)$ delay and clearly does not have any duplicate output.  This preprocessing phase takes $O(|D|)$ time and $O(|D|)$ space. We can now invoke Lemma~\ref{lem:helper:two} to achieve the claimed delay. The final observation is to express $J$ in terms of $|\tOUT_{\pi}|$. Note that $|\tOUT_{\pi}| \leq \Pi_{i \in [k]} |\domain(x_i)|$ which implies that $ \max_{i \in [k]} |\domain(x_i)| \geq |\tOUT_{\pi}|^{1/k}$. Thus, it holds that $J \geq |\tOUT_{\pi}|^{1/k}$ which gives us the desired bound on the delay guarantee.

\subsection{Dynamic setting}
{Remarkably, it is also possible to do a simple adaptation of the algorithm for the dynamic setting but only in the case of self-joins, where $R_1 = R_2 = \dots R_k = R$. In this setting, single tuple updates are allowed to the underlying relation, i.e., a tuple $t$ over variables $x, y$ can be either inserted or deleted from the relation $R(x, y)$ (note that a relation does not allow duplicates). 

\revb{We will store the relation $R(x,y)$ in a bidirectional hash table (i.e., given a $u$, we can retrieve $\pi_y \sigma_{x = u} R$ and $|\pi_y \sigma_{x = u} R|$, and given a $v$, we can retrieve $\pi_x \sigma_{y = v} R$ and $|\pi_x \sigma_{y = v} R|$). We will also use a max heap structure that admits constant amortized time insertion and updates (i.e., changing the weight of a key in the heap), and amortized $O(\log N)$ time to remove the element at the top of the heap wiht $N$ elements. Fibonacci heap~\cite{fredman1987fibonacci} is an example of such a data structure.

\begin{algorithm}[!htp]
	\SetCommentSty{textsf}
	\DontPrintSemicolon 
	\SetKwFunction{preprocess}{\textsc{Preprocess}}
        \SetKwFunction{update}{\textsc{Update}}
        \SetKwFunction{enumerate}{\textsc{Enumerate}}
        \SetKwFunction{enumeratel}{\textsc{EnumerateFixed$w_k$}}
	\SetKwInOut{Inp}{\textsc{input}}
	\SetKwProg{myproc}{\textsc{procedure}}{}{}
	\BlankLine
        \revb{
        $\mathsf{fulljoinsize}\leftarrow 0$, $J \leftarrow \emptyset$,
        $\mathsf{heap} \leftarrow \emptyset$; \tcc*{empty hash set and heap}
        \myproc{\preprocess{}}{
        \ForEach{$t \in  R$}{
            add $(\pi_{x} t, \pi_{x} t)$ to $J$ \;
            $\mathsf{fulljoinsize}\leftarrow \mathsf{fulljoinsize}+ |\sigma_{y = \pi_{y} t} R|^k $; \;
            add key $\pi_y t$ with weight $|\sigma_{y = \pi_{y} t} R|^k$ to $\mathsf{heap}$ \label{line:heapinsert} 
        }
        }
        \myproc{\update{$t, \mathsf{op}$}}{
            $\mathsf{oldcontribution} \leftarrow |\sigma_{y = \pi_{y} t} R|^k$ \label{line:temp}\;
            $\mathsf{newcontribution} \gets 0$ \;
            \If{$\mathsf{op} = \mathsf{insert}$}{
                insert $t$ in $R$; \tcc*{insert in bidirectional hash tables}
                add $(\pi_{x} t, \pi_{x} t)$ to $J$ \;
                $\mathsf{newcontribution} \gets |\sigma_{y = \pi_{y} t} R|^k$
            }
            \If{$\mathsf{op} = \mathsf{delete}$}{
                delete $t$ from $R$ \tcc*{delete from bidirectional hash tables}
                \lIf{$|\sigma_{x = \pi_x t} R| = 0$}{
                    delete $(\pi_x t, \pi_x t)$ from $J$}
            }
            $\mathsf{fulljoinsize}\leftarrow \mathsf{fulljoinsize}+ \mathsf{newcontribution} - \mathsf{oldcontribution}$; \label{line:new} \;
            update key $\pi_y t$ to weight $|\sigma_{y = \pi_{y} t} R|^k$ in $\mathsf{heap}$ \label{line:heapupdate} \;
            Rebuild $\mathsf{heap}$ by removing all keys with weight $0$ after every $O(|R|)$ updates 
            \label{line:rebuild}
        }
        \myproc{\enumerate{}}{
                $s \leftarrow $ key at the top of the $\mathsf{heap}$ \;
                Let $J'$ be the set with larger cardinality out of $J$ or $\times_{i \in [k]} \pi_{x_i} \sigma_{y = s} R_i$ \label{line:large} \;
                Call Lemma~\ref{lem:helper:two} with stored hash set $J'$, $\mA$ as any worst-case optimal join algorithm on the input query, and $T = c^\star \cdot \max\{\mathsf{fulljoinsize}, |\times_{i \in [k]} \pi_{x_i} \sigma_{y = s} R_i| \}$ \tcc*{$c^\star$ is the constant in $O(\cdot)$ of the worst-case optimal join algorithm}
        }}
	\caption{{\sc DynamicStarQuerySelfJoin}$(\qstar{k}, R)$} \label{algo:dynamic:star}
\end{algorithm}

Algorithm~\ref{algo:dynamic:star} shows the procedures for preprocessing, update, and enumeration. In the preprocessing phase, we populate the set $J$, which stores precomputed results that will be used for interleaving, and compute the full join size and store it in the \textsf{joinsize} variable. The hash maps storing the relation can be maintained in constant time under single-tuple updates for an update made to the relation. The set $J$ is also maintained by adding a new tuple whenever a new domain value for $x$ is seen and removing a tuple $(a,a)$ from $J$ whenever $|\sigma_{x=a}R|$ becomes zero. 

For the enumeration phase, we apply Lemma~\ref{lem:helper:two} with stored hash set $J$, \reva{since the hash set is a data structure that allows constant-delay enumeration of its contents and does not contain any duplicates}. We can now state the result formally.

\begin{thm} \label{thm:star:delay:alternate:update}
	For a self-join star query $\qstar{k}$ and any single relation $R$, there exists an algorithm with preprocessing time $O(|R|)$ and preprocessing space $O(|R|)$, such that $\qstar{k}(R)$ can be enumerated with
	$
	\text{delay } \delta = O\bigg({|\tOUT_{\Join}|}/{|\tOUT_{\pi}|^{1/k}}\bigg) 
	\text{and enumeration space } S_e = O(|R|)
	$
	with $O(\log |R|)$ amortized update time for single-tuple updates.
\end{thm}
\begin{proof}
    It is easy to see that the procedure \textsc{Preprocess} requires $O(|R|)$ time and correctly computes the full join size $\mathsf{fulljoinsize}$ as $\sum_{v \in \domain(y)} |\sigma_{y = v} R|^k$. Each $y$ valuation in the active domain is inserted into the heap, taking constant  time per insertion. Each operation in the procedure \textsc{Update} also requires constant time by our assumption on the insertion and deletion time complexity of the hash table and the hash set. To maintain the full join size correctly, we need to consider the new degree of the $y$ valuation after inserting/deleting $t$ in $R$. Line~\ref{line:new} does exactly that by subtracting the contribution of the $y$ valuation before the update and adding the new contribution of the $y$ valuation to the full join size after the update. Since we also need to update the weight of $\pi_{y} t$ in the heap, line~\ref{line:heapupdate} updates the weight of the key appropriately in constant time. Finally, procedure \textsc{Enumerate} invokes Lemma~\ref{lem:helper:two}. The main idea here is to pick the larger of $J$ or $\left|\times_{i \in [k]} \pi_{x_i} \sigma_{y = s} R_i\right|$ so that the enumeration delay is smaller by larger output that can be interleaved with the query execution. Finding the value $s$ (which is at the top of the heap) is a constant time operation. Applying the same reasoning as Theorem~\ref{thm:star:delay:alternate}, we can get desired delay $\delta = O({|\tOUT_{\Join}|}/{|\tOUT_{\pi}|^{1/k}})$.

    To reason about the amortized update time guarantee, note that on line~\ref{line:rebuild}, we rebuild the \textsf{heap} after every $O(|R|)$ updates operations. The step is important to make sure that the space usage of the heap remains linear in the size of the input. The rebuilding step repeatedly pops all entries from the heap and then creates the heap with only the entries that have non-zero weight. Thus, the step requires $O(|R| \log |R|)$ time but since it is done only once after every $O(|R|)$ update operations, the cost amortized cost is $O(\log|R|)$. Since all other update operations are constant time, the heap rebuilding dominates the overall update guarantee.
\end{proof}

To achieve constant update time, the main result that applies to all hierarchical queries (including self-joins) from~\cite{kara19} requires linear preprocessing time and achieves worst-case linear delay. Since ${|\tOUT_{\Join}|}/{|\tOUT_{\pi}|^{1/k}} \leq |D|$~\cite{amossen2009faster}, our strategy provides an alternate algorithm that achieves the same worst-case linear delay (with $O(\log |D|)$ update time instead of constant update time) but may perform better on certain database instances. 

\paragraph{Optimality for dynamic self-joins} Next, we show that the delay guarantee obtained by the algorithm above is optimal conditioned on the popular Boolean matrix multiplication (BMM) conjecture~\cite{williams2018subcubic}: In the Word RAM model with words of $O(\log n)$ bits, given two $n \times n$ Boolean matrices, there exists no combinatorial algorithm for multiplying the matrices in $O(n^{3 - \epsilon})$ time for any constant $\epsilon > 0$. BMM conjecture is also hard when both input matrices are identical\footnote{Given two $n \times n$ matrices $M_1$ and $M_2$, we can create a matrix \begin{equation*}
        M = \left[
            \begin{array}{c c}
                \mathbf{0} & M_1 \\ 
                M_2 & \mathbf{0}
            \end{array}
        \right]
    \end{equation*} to show that if multiplication of $M$ with itself can be achieved in $O(n^{3 - \epsilon})$ time, then BMM can also be solved in $O(n^{3 - \epsilon})$ time. Here $\mathbf{0}$ is an $n \times n$ matrix with all $0$ entries.}. Although the term combinatorial is not formally defined, it refers to an algorithm that does not use any algebraic techniques (such as the ones used for fast matrix multiplication).

\begin{lem}
    For a self-join star query $\qstar{2}$, any single relation $R$ and parameter $\delta = O({|\tOUT_{\Join}|}/{\sqrt{|\tOUT_{\pi}|}})$, there exists no combinatorial algorithm with preprocessing time $O(|R|)$ and space $O(|R|)$, such that the query result $\qstar{2}(R)$ can be enumerated with delay $O(\delta^{1-\epsilon})$, for any constant $\epsilon > 0$, with $O(\log |R|)$ amortized update time for single-tuple updates, assuming the Boolean matrix multiplication conjecture.
\end{lem}
\begin{proof}
    Consider the input matrix $M$ for the BMM problem where we want to compute the product of $M$ with itself. As shown below, we construct a matrix $M'$ of size $2n \times 2n$.
    \begin{equation*}
        M' = \left[
            \begin{array}{c c}
                \mathbf{0} & M \\ 
                \mathbf{1} & \mathbf{0}
            \end{array}
        \right]
    \end{equation*}

    Here, \textbf{0} and \textbf{1} represent a $n \times n$ matrix with all entries as $0$ and $1$ respectively. We will now create a relation $R(x,y)$ from $M'$ by calling the \textsc{Update} procedure for every non-zero entry $(i,j)$ of $M'$. Since there are $O(n^2)$ entries and every entry is processed in $O(\log n)$ amortized time, the total time required is $O(n^2 \log n)$. Next, observe that due to the \textbf{1} matrix in the lower left half of $M'$, there is a $y$ value (in fact, there are $n$ such $y$ values) that generates $\Theta(n^2)$ amount of output. Also, note that $|\tOUT_{\Join}| = \Theta(n^3)$. Therefore, Algorithm~\ref{algo:dynamic:star} would provide a delay guarantee of $O(n^3/n^2) = O(n)$. Suppose an algorithm could enumerate the result with $O(n^{1-\epsilon})$ delay. The total output of the join can be as large as $O(2n \cdot 2n) = O(n^2)$. The original output of multiplying $M$ with itself can be recovered by only looking at output tuples of the join result $(z,w)$ where both $z \leq n$ and $w \leq n$. Thus, the entire output can be enumerated in total time $O(n^2 \cdot n^{1-\epsilon}) = O(n^{3 - \epsilon})$ time, which contradicts the BMM conjecture.
\end{proof}

\paragraph{Why non-self-join is hard.} At this point, the reader may wonder what breaks down for the non-self-join case. First, note that our lower bound argument directly extends to the non-self-join case as well. However, in Algorithm~\ref{algo:dynamic:star}, we can no longer maintain $J$ under updates in the non-self-join case. Indeed, in the presence of self-joins, each $(a,b) \in R$ guarantees that $(a,a, \dots, a)$ is a part of the output. However, no such claim can be made when the relations are independent. Algorithm~\ref{algo:dynamic:star}  keeps track of the $y$ valuation that generates the largest amount of output by using the heap and use it for interleaving (see line~\ref{line:large} where pick the larger of $J$ and the $y$ valuation that generates the most output), but in that case, we do not get any closed form output dependent expression for the delay guarantee. The observation that self-joins can make the  evaluation of queries easier (but in the static setting) has also been made by prior work~\cite{carmeli2023conjunctive}. Thus, constructing an algorithm for the non-self-join that leads to an output-sensitive guarantee is a non-trivial problem that likely requires the development of new technical machinery.}}

\subsection{Fast Matrix Multiplication}
\label{sec:fmm}

{So far, we have focused on combinatorial algorithms.  We next show how fast matrix multiplication can obtain a tradeoff between preprocessing time and delay better than Theorem~\ref{thm:basic:three} for some delay values.
	
\begin{thm} \label{thm:twopath:mmul}
	For a star query $\qstar{k}$ and any instance $D$, there exists an algorithm with pre-processing time $T_p = O((|D|/\delta)^{\omega+k-2})$ and space $S_p=O((|D|/\delta)^{\omega+k-2})$ such that the query result $\qstar{k}(D)$ can be enumerated with delay $O(\delta)$ and enumeration space $S_e = O(|D|)$, for $1 \leq \delta \leq |D|^{(\omega + k - 3)/(\omega + 2 \cdot k - 3)}$.
\end{thm}
\begin{proof}
	Let $\delta$ be the degree threshold for deciding whether a valuation is heavy or light. We can partition the original query into the subqueries of the form:
	$$\pi_{x_1, \dots, x_k} (R_1(x_i^?, y^?) {\Join} R_2(x_2^?, y^?) {\Join} \dots {\Join} R_k(x_k^?, y^?))$$
	where $?$ can be either $h, \ell$ or $\star$. \revb{We will reuse the definition of heavy and light values for variables $x_i$ and $y$ as defined in Section~\ref{sec:main}. The first observation is that the result of the $k$ subqueries where superscript of  $x_1 \dots x_{i-1}$ is $h$, superscript for $x_i$ is $\ell$, and it is $\star$ for the remaining variables (including $y$), can be enumerated with $O(\delta)$ delay using the \textsc{ListMerge} subroutine for each valuation $v$ of $x^\ell_i$ that is light. Partitioning of each relation $R_i$ into $R_i(x^h_i, y^\star)$ and $R_i(x^\ell_i, y^\star)$, as well as their indexing into a bidirectional hash table, can be done in linear time. These $k$ queries can be directly enumerated in the enumeration phase with the desired delay.
    
    For the remaining queries, we need to perform preprocessing. The first subquery to process is $Q'$ with superscript for all $x_i, i \in [k]$ is $\star$ and it is $\ell$ for $y$. We define a valuation $v$ for variable $y$ to be light if $|\pi_{x_i} \sigma_{y = v} R(x_i, y)| \leq \delta$ for each $i \in \{ 1, \dots, k\}$. Using any worst-case optimal join algorithm, this subquery can be evaluated in time $ O(\sum_{\text{all light valuations }v} |\pi_{x_i} \sigma_{y = v} R(x_i, y)|) \leq O(\delta^{k-1} \cdot \sum_{\text{all light valuations }v} |\pi_{x_i} \sigma_{y = v} R(x_k, y)|) = O(|D| \cdot \delta^{k-1})$.

    The second query we need to preprocess is the subquery $Q''$, where all variables are heavy, i.e., superscript for $x_i, i \in [k]$ and $y$ is $h$. For this subquery, we will use matrix multiplication using the construction outlined in~\cite{deep2020fast}. Specifically, we first create a rectangular matrix $M$ of dimensions $(|D|/\delta)^{\lceil k/2 \rceil} \times |D|/\delta$ and a matrix $N$ of dimensions $|D|/\delta \times (|D|/\delta)^{\lfloor k/2 \rfloor}$ such that,

    \begin{equation*} 
        M_{(a_1, a_2, \dots, a_{\lceil k/2 \rceil}), b} = \begin{cases}
            1, &(a_1,b) \in R_1, \dots, (a_{\lceil k/2 \rceil}, b) \in R_{\lceil k/2 \rceil} \\
            0, &\text{otherwise}
        \end{cases}
    \end{equation*}

    \begin{equation*} 
        \hspace{2.4em}N_{b,(a_{\lceil k/2 \rceil + 1}, \dots, a_{k})} = \begin{cases}
            1, &(a_{\lceil k/2 \rceil + 1},b) \in R_{\lceil k/2 \rceil + 1}, \dots, (a_{k}, b) \in R_{k} \\
            0, &\text{otherwise}
        \end{cases}
    \end{equation*}

    Observe that the non-zero entries in the product of matrices $M$ and $N$ correspond to the output of the subquery $Q''$. In other words, a non-zero entry for row $(a_1, a_2, \dots, a_{\lceil k/2 \rceil})$ and column $(a_{\lceil k/2 \rceil + 1}, \dots, a_{k})$ implies that there exists a value $b$, such that $(a_i, b) \in R_i$ for each input relation and thus $(a_1, a_2, \dots, a_k) \in Q''$. By construction, it is also guaranteed that every output tuple in the result of $Q''$ will have a non-zero entry in the product of $M$ and $N$. Using Lemma~\ref{lem:matrix:multiplication}, we can construct and multiply the matrices in time $O((|D|/\delta)^{\lceil k/2 \rceil + 1 + \lfloor k/2 \rfloor} \cdot (|D|/\delta)^{\omega-3}) = O((|D|/\delta)^{\omega+k-2})$ since the smallest dimension $\beta = (|D|/\delta)$ as $k \geq 2$. The output of $Q'$ and $Q''$ can be deduplicated using a hash set.

    We are ready to enumerate the query result. The key observation is that the materialized and deduplicated output of $Q'$ and $Q''$ can be enumerated with constant delay, each of $k$ subqueries where $x_i = \ell$ can be enumerated with delay $O(\delta)$, and there is no overlap in the output between any of the queries during the enumeration phase. It remains to argue the total preprocessing time required for processing $Q'$ and $Q''$. Processing the two queries requires $T_p = O(|D| \cdot \delta^{k-1} + (|D|/\delta)^{\omega+k-2})$. By ensuring that $|D| \cdot \delta^{k-1} \leq (|D|/\delta)^{\omega+k-2}$, we get the desired range of $ 1 \leq \delta \leq |D|^{(\omega + k - 3)/(\omega + 2 \cdot k - 3)}$.}  
    \eat{ 
	\begin{itemize}
		\item $x_i$ has $? = \ell${, $y$ has $?=\star$} and $z$ has $? = \star$. In this case, we can invoke {\sc ListMerge}$(v)$ for each valuation $v$ of attribute $x_i$ and enumerate the output. Note that 
		\item $x$ has $? = h${, $y$ has $?=\star$} and $z$ has $? = \ell$. In this case, we can invoke {\sc ListMerge}$(v_i)$ for each valuation $v_i$ of attribute $z$ and enumerate the output. Note that there is no output overlap between this case and the previous one.
		\item both $x,z$ have $? = h$. We compute the output of $\pi_{x,z} R(x^h, y^?) \Join S(y^?, z^h)$ in preprocessing phase and obtain $O(1)$-delay enumeration. {In the following, we say that $y$ has $? = \ell$ to mean that the join considers all $y$ valuations that have degree at most $\delta$ in both $R$ and $S$.}
		\begin{itemize}
			\item $y$ has $? = \ell$. We compute the full join $R(x^h, y^\ell) \Join S(y^\ell, z^h)$ and materialize all distinct output results, which takes $O(|D| \cdot \delta)$ time. 
			\item $y$ has $? = h$. All attributes have at most $|D|/\delta$ valuations. We now have a square matrix multiplication instance where all dimensions have size $O(|D|/\delta)$. Using Lemma~\ref{lem:matrix:multiplication}, we can evaluate the join in time $O((|D|/\delta)^{\omega})$. 
		\end{itemize}
	\end{itemize}
	
	Overall, the preprocessing time is $T_p = O((|D|/\delta)^\omega + |D| \cdot \delta)$. The matrix multiplication term dominates whenever $\delta \leq O(|D|^{(\omega-1)/(\omega+1)})$ which gives us the desired time-delay tradeoff.}
\end{proof}
  For the two-path query and the current best value of $\omega = 2.3715$~\cite{williams2024new}, we get the tradeoff as $T_p = O((|D|/\delta)^{2.3715})$ and a delay guarantee of $O(\delta)$ for $|D|^{0.15} < \delta \leq |D|^{0.406}$. If we choose $\delta = |D|^{0.40}$, the worst-case preprocessing time is $T_p = O(|D|^{1.4229})$. In contrast, Theorem~\ref{thm:basic:three} requires a worst-case preprocessing time of $T_p = O(|D|^{1.6})$, which is suboptimal compared to the above theorem. On the other hand, since $T_p = O(|D|^{1.4229})$, \changes{we can safely assume that $|\tOUT_{\Join}| > |D|^{1.4229}$, otherwise one can simply compute the full join in time $c^\star \cdot |D|^{1.4229}$ using {\sc ListMerge}, deduplicate and get constant delay enumeration.} Applying Theorem~\ref{thm:star:delay} with $|\tOUT_{\Join}| > |D|^{1.4229}$ tells us that we can obtain delay as $O(|D|^2/ |\tOUT_{\Join}|) = O(|D|^{0.5771})$. Thus, we can offer the user both choices, and the user can decide which algorithm to use.	
\section{Left-Deep Hierarchical Queries}
\label{sec:ldeep}

In this section, we will apply our techniques to another subset of hierarchical queries, which we call {\em left-deep}. 
A left-deep hierarchical query is of the following form:
\[
Q_{\texttt{leftdeep}}^k = \pi_{w_1,w_2,\cdots,w_k}R_1(w_1, x_1) {\Join} R_2(w_2, x_1, x_2) {\Join} \dots {\Join} R_{k}(w_{k}, x_1, \dots, x_{k})
\]

\begin{algorithm}[t]
	\SetCommentSty{textsf}
	\DontPrintSemicolon 
	\SetKwFunction{preprocess}{\textsc{Preprocess}}
        \SetKwFunction{enumerate}{\textsc{Enumerate}}
        \SetKwFunction{enumeratel}{\textsc{EnumerateFixed$w_k$}}
	\SetKwInOut{Inp}{\textsc{input}}
	\SetKwProg{myproc}{\textsc{procedure}}{}{}
	\BlankLine
        \revb{
        \myproc{\preprocess}{
        Compute $|\tOUT_{\Join}|$ \;
        Fix $\delta = 2 \cdot |D|^k / |\tOUT_{\Join}|$ \;
        Apply the domain compression transformation using $f$ and $f^{-1}$ on each relation as described in~\autoref{sec:warmup} \;
        For relation $R_k$, create two additional indexes: (i) using a hash table $H^\ell_k$ with key as all light valuations $v$ for variable $w_i$ and value as the list $\pi_{x_1, \dots, x_i}\sigma_{w_i = v} R_i$; (ii) using a hash table $H^h_k$ with key as all heavy valuations $v$ for variable $w_i$ and value as the list $\pi_{x_1, \dots, x_i}\sigma_{w_i = v} R_i$ \;
        Create an index (called $H_i(u)$) for each input relation $R_i, i \in \{1, \dots, k-1\}$ using a hash table with key as valuation $u$ for variables $x_1, \dots, x_i$ and value as the sorted list $\pi_{w_i}\sigma_{x_1 = u[x_1], \dots, x_i = u[x_i]} R_i$
        }
        \myproc{\enumeratel{$v$}}{
            Consider each $u \in \pi_{x_1, \dots, x_k}\sigma_{w_k = v} R_k$. Let $\mA^v_u$ be the algorithm that enumerates $\times^{k-1}_{j=1} (H_j(\pi_{x_1, \dots, x_j} u))$ in lexicographic order with constant delay \;
            Invoke \autoref{algo:merging} with input as $\mA^v_u, u \in \pi_{x_1, \dots, x_k}\sigma_{w_k = v} R_k$. Capture the output $w$ emitted on~\autoref{line:output} of \autoref{algo:merging}.\;
            \textbf{output} $f^{-1}(v, w)$
        }
        \myproc{\enumerate}{
            Let $\mA'$ be the algorithm that calls \enumeratel{$v$} for all light $w_k$ valuations, requiring total time $T' = |\tOUT_{\Join}|/2$ \;
            Let $\mA$ be the algorithm that calls \enumeratel{$v$} for all heavy $w_k$ valuations, requiring total time $T = c^\star \cdot |\tOUT_{\Join}|/2$ \tcc*{$c^\star$ is the upper bound on the number operations required by one invocation of \enumeratel{$v$}}
            Apply Lemma~\ref{lem:helper:one} with $\mA', \mA, T, T'$ as input
        }}
	\caption{{\sc LeftDeepHierarchical}$(Q_{\texttt{leftdeep}}^k, D)$} \label{algo:leftdeep}
\end{algorithm}

\noindent $Q_{\texttt{leftdeep}}^k$ is a hierarchical query for any $k \geq 1$. For $k=2$, we get the two-path query. For $k=3$, we get $\pi_{x_1,x_2,x_3}R_1(w_1, x_1) \Join R_2(w_2, x_1, x_2) \Join R_3(w_3, x_1, x_2, x_3)$ \revb{(hypergraph shown in Figure~\ref{fig:hyp})}.

\begin{figure}[!htp]
\begin{tikzpicture}
[
    he/.style={draw, rounded corners},        
    ce/.style={draw, dashed, rounded corners} 
]

\node (e) at (2,0) {$x_3$};
\node (f) at (2,-1) {$w_3$};
\node (d) at (2,1) {$x_2$};
\node (a) at (2,2) {$x_1$};
\node (b) at (3,2) {$w_1$};
\node (c) at (1,1) {$w_2$};

\node [he, fit=(a) (b)] {};
\node [label=$R_1$] at (4,1.625) {};
\node [he, fit=(a) (d) (e) (f)] {};
\node [label=$R_2$] at (0,1) {};
\node [he, fit=(a) (d) (c)] {};
\node [label=$R_3$] at (3,-1.5) {};

\end{tikzpicture}

    \caption{Hypergraph for query $Q_{\texttt{leftdeep}}^3$.}
    \label{fig:hyp}
\end{figure}

\revb{\autoref{algo:leftdeep} shows the preprocessing and enumeration phases for left-deep hierarchical queries. Let us first look at the preprocessing step. Similar to star queries, the algorithm begins by computing the size of the full join and performing the domain transformation to make sure all values in the relation as between $1$ and $|D|$. Then, we create a bidirectional index for relation $R_k$. For each valuation $v$ of $w_k$, we store the list $\pi_{x_1, \dots x_k}\sigma_{w_k = v} R_k$. The light valuations have the list stored in the hash table $H^\ell_k$ and the heavy valuations have the lists stored in $H^h_k$. For the remaining relations, we store a sorted list of $w_i$ for each tuple $\pi_{x_1, \dots, x_i} R_i$ in the hash table $H_i$.

During the enumeration phase, like star queries, we will interleave the algorithm's execution that processes the light $w_k$ valuations with the heavy $w_k$ valuations. In particular, procedure \textsc{EnumerateFixed$w_k$($v$)} shows how to generate the output for a fixed $w_k$ valuation $v$. In particular, for any fixing $u$ of $\pi_{x_1, \dots x_k}\sigma_{w_k = v} R_k$, the cartesian product of the $H_j(\pi_{x_1, \dots x_k}) u)$ can be enumerated with constant delay in lexicographic order. Therefore, now we can use the \textsc{Merge} subroutine directly (Algorithmn~\ref{algo:merging}) to combine the output for all $u \in \pi_{x_1, \dots x_k}\sigma_{w_k = v} R_k$. Since only the light $w_k$ valuations will give us a small delay guarantee, the interleaving with the heavy valuations will help us ensure we can get a possibly sublinear overall guarantee.} We show that the following result holds:

\begin{thm}
	For a left-deep query $Q_{\texttt{\upshape leftdeep}}^k$ and any instance $D$, there exists an algorithm with preprocessing time $T_p = O(|D|)$ and preprocessing space $S_p = O(|D|)$ such that the query results $Q_{\texttt{leftdeep}}^k(D)$ can be enumerated with delay $O(|D|^{k}/|\tOUT_{\Join}|)$ and enumeration space $S_e = O(|D|)$. \label{thm:main:leftdeep}
\end{thm}
\begin{proof}
\revb{
    Let the degree threshold be $\delta = 2 \cdot |D|^k / |\tOUT_{\Join}|$. First, observe that each procedure step \textsc{Preprocess} requires only $O(|D|)$ time. Indeed, computing $|\tOUT_{\Join}|$, domain compression, and index creation all require a linear pass on each relation $R_i$. Next, we argue that procedure \textsc{EnumeratedFixed$w_k$($v$)} can enumerate output with delay $O(|\pi_{x_1, \dots, x_k}\sigma_{w_k = v} R_k|)$. First, note that for a fixed $u \in \pi_{x_1, \dots, x_k}\sigma_{w_k = v} R_k$, $\mA^v_u$ enumerates the cartesian product with constant delay. Thus, since we invoke $|\pi_{x_1, \dots, x_k}\sigma_{w_k = v} R_k|$ instances of Algorithm~\ref{algo:merging}, one for each value of $u$, each of which enumerates sorted results in constant delay, the overall delay is $O(|\pi_{x_1, \dots, x_k}\sigma_{w_k = v} R_k|)$ according to Lemma~\ref{lem:helper:three}. 

    First, we analyze the total time $T$ required by $\mA$ that processes all heavy valuations. As the degree threshold is $\delta$, there are at most $|D|/\delta$ heavy $w_k$ valuations, and each of them may generate $|D|^{k-1}$ output for the $k-1$ remaining projection variables. Therefore, $T \leq c^\star \cdot |D|^{k-1} \cdot |D| /\delta = c^\star \cdot |D|^k/\delta = c^\star \cdot |\tOUT_{\Join}|/2$, where $c^\star$ is the upper bound on the number of operations performed by \textsc{EnumeratedFixed$w_k$($v$)} for generating one output tuple. Further, if $Q_H$ and $Q_L$ denote the subqueries without projections evaluated by $\mA$ and $\mA'$, it holds that $|Q_H| + |Q_L| = |\tOUT_{\Join}|$ since both algorithms are doing list merging and go over every tuple in the full join once. Therefore, it holds that $|Q_L| = |\tOUT_{\Join}| - |Q_H| = |\tOUT_{\Join}| - |D|^k/\delta \geq |\tOUT_{\Join}|/2$. Therefore, a lower bound on $T'$ (i.e., time required to process all light valuations) is $|\tOUT_{\Join}|/2$.

	We can now apply Lemma~\ref{lem:helper:one} with the following parameters
\begin{itemize}
     \item $\mA'$ with $T' = |\tOUT_{\Join}|/2$ and delay guarantee of $\mA'$ as $O(\delta)$
     \item $\mA$ with $T = c^\star \cdot |\tOUT_{\Join}|/2$
     \item $T,T'$ are fixed once $|\tOUT_{\Join}|$ and constants $c^\star$ has been determined by analyzing the program.
\end{itemize}

\noindent Lemma~\ref{lem:helper:one} implies the delay $O(\delta \cdot \max\{1, T/T'\}) = O(\delta)$, achieving the desired guarantee.}
\end{proof}

In the above theorem, $\tOUT_{\Join}$ is the full join size of $Q_{\texttt{leftdeep}}^k$. The AGM exponent for $Q_{\texttt{leftdeep}}^k$ is $k$.
Observe that Theorem~\ref{thm:main:leftdeep} is of interest when $|\tOUT_{\Join}| > |D|^{k-1}$ to ensure that the delay is smaller than $O(|D|)$. When the condition $|\tOUT_{\Join}| > |D|^{k-1}$ holds, the delay obtained by Theorem~\ref{thm:main:leftdeep} is also better than the one given by the tradeoff in Theorem~\ref{thm:basic:three}. In the worst-case when $|\tOUT_{\Join}| = \Theta(|D|^k)$, we can achieve constant delay enumeration after linear preprocessing time, compared to Theorem~\ref{thm:basic:three} that would require $\Theta(|D|^k)$ preprocessing time to achieve the same delay. \changes{The decision of when to apply Theorem~\ref{thm:main:leftdeep} or Theorem~\ref{thm:basic:three} can be made in linear time by checking whether $|D|^k/|\tOUT_{\Join}|$ is smaller or larger than the actual delay guarantee obtained by the algorithm of Theorem~\ref{thm:basic:three} after linear time preprocessing.}
	\section{Path Queries}
\label{sec:path}

In this section, we will study path queries $P_k$. 
Recall that for $k \geq 3$, $P_k$ is not a hierarchical query; hence, the tradeoff from~\cite{kara19} does not apply. {A subset of path queries, namely 3-path and 4-path counting queries, were considered in~\cite{kara2019counting}. The algorithm used for counting the answers of 3-path and 4-path queries under updates constructed a set of views that can be used to enumerate the query results under the static setting. Our result extends the same idea to apply to arbitrary length path queries, which we state next.}

\begin{thm} \label{thm:path}
For a path query $P_k$, any instance $D$ and parameter $\epsilon \in [0,1)$, there exists an algorithm with preprocessing time $T_p = O(|D|^{2 - \epsilon/(k-1)})$ and preprocessing space $S_p = O(|D|^{2 - \epsilon/(k-1)})$, such that the query result $P_k(D)$ can be enumerated with delay $O(|D|^\epsilon)$ and enumeration space $S_e = O(|D|)$.
\end{thm}
\begin{proof}
	Let $\Delta$ be a parameter that we will fix later. In the preprocessing phase, we first perform a full reducer pass to remove dangling tuples, {apply the domain transformation technique by creating $f$ and $f^{-1}$} and then create a hash map for each relation $R_i(x_i, x_{i+1})$ with key $x_i$, and all its corresponding $x_{i+1}$ values sorted for each key entry. (We also store the degree of each value.) Next, for every $i=1, \dots, k$, and every heavy value $a$ of $x_{i}$ in $R_i$ (with degree $> \Delta$), we compute the query $\pi_{x_{k+1}}(R_i(a, x_{i+1}) \Join \dots \Join R_k(x_k, x_{k+1}))$, and store its result sorted in a hash map with key $a$. Note that each such query can be computed in time $O(|D|)$ through a sequence of semijoins and projections, { and sorting in linear time using count sort}. Since there are at most $|D|/\Delta$ heavy values for each $x_i$, the total running time (and space necessary) of this step is $O(|D|^2/\Delta)$.
	
	We will present the enumeration algorithm using induction. In particular, for each $i=k, \dots, 1$ and for every value $a$ of $x_i$, the subquery $\pi_{x_{k+1}}(R_i(a, x_{i+1}) \Join \dots \Join R_k(x_k, x_{k+1}))$ can be enumerated (using the same order) with delay $O(\Delta^{k-i} )$. This implies that our target path query can be enumerated with delay $O(\Delta^{k-1})$ by simply iterating through all values of $x_1$ in $R_1$. Finally, we can obtain the desired result by choosing $\Delta = |D|^{\epsilon/(k-1)}$. Indeed, for the base case ($i=k$) it is trivial to see that we can enumerate $\pi_{x_{k+1}}(R_k(a, x_{k+1}))$ with constant delay $O(1)$ using the stored hash map. Consider some $i$ and a value $a$ for $x_i$ in $R_i$ for the inductive step. If the value $a$ is heavy, we can enumerate all the $x_{k+1}$'s with constant delay by probing the hash map we computed during the preprocessing phase. If the value is light, then there are at most $\Delta$ values of $x_{i+1}$. For each such value $b$, the inductive step provides an algorithm that enumerates all $x_{k+1}$ with delay $O(\Delta^{k-i-1})$. Observe that the order across all $b$'s will be the same. Thus, we can apply Lemma~\ref{lem:helper:three} to obtain that we can enumerate the union of the results with delay $O(\Delta \cdot \Delta^{k-i-1}) = O(\Delta^{k-i})$. Finally, For each output tuple $t$ generated, we return $f^{-1}(t)$ to the user.
\end{proof}

When $\epsilon =1$, we can obtain a delay $O(|D|)$ using only linear preprocessing time $O(|D|)${ using the result of~\cite{bagan2007acyclic} since the query is acyclic}, while when $\epsilon$ approaches to $1$,  the above theorem would give preprocessing time $O(|D|^{2-1/(k-1)})$. Hence, for $k \geq 3$, we observe a discontinuity in the time-delay tradeoff. A second observation following from Theorem~\ref{thm:path} is that as $k \rightarrow \infty$, the tradeoff collapses to only two extremal points: one where we get {constant} delay with $T_p = O(|D|^2)$, and the other where we get linear delay with $T_p = O(|D|)$.

	\section{Related Work} \label{sec:related}

We overview prior work on static query evaluation for acyclic join-project queries. The result of any acyclic CQ can be enumerated with constant delay after linear time
preprocessing if and only if it is free-connex~\cite{bagan2007acyclic}. This is based on the conjecture that Boolean multiplication of $n \times n$ matrices cannot be done in $O(n^2)$ time. Acyclicity itself is necessary for having constant delay enumeration: A CQ admits constant delay enumeration after linear-time preprocessing if and only if it is free-connex acyclic~\cite{brault2013pertinence}. This is based on a stronger hypothesis that the existence of a triangle in a hypergraph of $n$ vertices cannot be tested in
time $O(n^2)$ and that for any $k$, testing the presence of a $k$-dimensional tetrahedron cannot be tested in linear time. We refer readers to an overview of problems and progress related to constant delay enumeration in~\cite{segoufin2015constant}.  Prior work also exhibits a dependency between the space and enumeration delay
for conjunctive queries with access patterns~\cite{deep2018compressed}. It constructs a succinct representation of the query result that allows for the enumeration of tuples over some variables under value bindings for all other variables. As noted by~\cite{kara19}, it does not support enumeration for queries with free variables, which is also its main contribution. Our work demonstrates that for a subset of hierarchical queries, the tradeoff shown in~\cite{kara19} is not optimal. Our work introduces fundamentally new ideas that may be useful in improving the tradeoff for arbitrary hierarchical queries and enumeration of UCQs. There has also been some experimental work by the database community on problems related to enumerating join-project query results efficiently but without any formal delay guarantees. Seminal work~\cite{graphgen2015, graphgen2017, graphgen2017adaptive, deep2019ranked} has studied how compressed representations can be created apriori that allow for faster enumeration of query results. {For the two path query, the fastest evaluation algorithm (with no delay guarantees) evaluates the projection join output in time $O(|D| \cdot |\tOUT_{\pi}|^{\frac{(\omega-1)}{(\omega+1)}}+ |D|^{\frac{2(\omega-1)}{(\omega+1)}} \cdot |\tOUT_{\pi}|^{\frac{2}{(\omega+1)}})$~\cite{deep2020fast,amossen2009faster}. For star queries, there is no closed form expression but  fast matrix multiplication can be used to obtain instance dependent bounds on running time.} Also related is the problem of dynamic evaluation of hierarchical queries. Recent work~\cite{kara2019counting, kara19, berkholz2017answering,berkholz2018answering} has studied the tradeoff between amortized update time and delay guarantees. Some of our techniques may also lead to new insights and improvements in existing algorithms. Prior work in differential privacy~\cite{roy2020crypt} and directed graphical models~\cite{chowdhury2020data} may also benefit from some of our techniques.	
	\section{Conclusion and Open Problems} \label{sec:conclusion}
In this paper, we studied the problem of enumerating query results for an important subset of CQs with projections, namely star and path queries. We presented data-dependent algorithms that improve upon existing results by achieving non-trivial delay guarantees in linear preprocessing time. Our results are based on interleaving join query computation to achieve meaningful delay guarantees. Further, we showed how non-combinatorial algorithms (fast matrix multiplication) can be used for faster preprocessing to improve the tradeoff between preprocessing time and delay. We also presented new results on time-delay tradeoffs for a subset of non-hierarchical queries for the class of path queries.  Our results also open several new tantalizing questions that open up possible directions for future work.

\begin{itemize}
    \item {\em More preprocessing time for star queries.} One major open question is whether Theorem~\ref{thm:star:delay} can benefit from more preprocessing time to achieve lower delay guarantees. For instance, if we can afford the algorithm preprocessing time $T_p = O(|\tOUT_{\Join}|/|D|^{\epsilon} + |D|)$ time, can we expect to get delay $\delta = O(|D|^\epsilon)$ for all $\epsilon \in (0,1)$?
    \item {\em Sublinear delay guarantees for the two-path query.} It is not known whether we can achieve a sublinear delay guarantee after linear preprocessing time for $Q_{\texttt{two-path}}$. This question is equivalent to the following problem: for what values of $|\tOUT_{\pi}|$ can $Q_{\texttt{path}}$ be evaluated in linear time? If $|\tOUT_{\pi}| = |D|^\epsilon$, the best known algorithms can evaluate $Q_{\texttt{two-path}}$ in time $O(|D|^{1+\epsilon/3})$ (using fast matrix multiplication)~\cite{deep2020fast} but this is still superlinear.
    \item {\em Space-delay bounds.} The last question is to study the tradeoff between space vs delay for arbitrary hierarchical and path queries. Using some of our techniques, it may be possible to smartly materialize a certain subset of joins that could be used to achieve delay guarantees by interleaving with join computation. We also believe the space-delay tradeoff implied by prior work can be improved for specific delay ranges using the ideas introduced in this paper.
\end{itemize}

	\newpage
	\bibliographystyle{alphaurl}
	\bibliography{reference}
	
\end{document}